\documentclass[11pt,reqno]{amsart}

\usepackage{amssymb,amsmath,amsthm,graphics,epsfig}
\usepackage{amsfonts}
\usepackage{graphicx,cite,caption}
\usepackage{amssymb,tabularx}
\usepackage{lineno}
\usepackage{etoolbox}

\usepackage{multirow} 
\usepackage{array} 
\usepackage{graphicx} 
\usepackage{caption} 
\usepackage[table]{xcolor} 
\definecolor{lightgrayrow}{RGB}{240,240,240}

\usepackage{makecell}

\makeatletter
\pretocmd{\caption}{\linenomath}{}{}
\apptocmd{\caption}{\endlinenomath}{}{}
\makeatother
\usepackage{mathrsfs}
\usepackage[colorlinks=true, pdfstartview=FitV, linkcolor=blue, citecolor=blue, urlcolor=blue]{hyperref}
\usepackage{booktabs}
\usepackage{graphicx}
\usepackage{tikz}
\usepackage{subfigure}
\usepackage{cases}
\usepackage{booktabs}
\usepackage{caption}
\usepackage{orcidlink}
\usepackage{comment} 
\usepackage{arydshln}
\usepackage[colorinlistoftodos]{todonotes}
\usepackage[normalem]{ulem}
\usepackage{color}
\usepackage{amsaddr} %direcciones adelante
\usepackage[margin=2.7cm]{geometry}

\newtheorem{Theorem}{Theorem}
\newtheorem{Proposition}{Proposition}
\newtheorem{Corollary}{Corollary}
\usepackage{float}
\theoremstyle{definition}

\newtheorem{Remark}{Remark}

\definecolor{ColorEdward}{rgb}{0.4,0.6,0.7}

\usepackage{interval}

\date{\today}

\begin{document}

%\title{Simple quasispecies models with time lags in RNA replicase-mediated complementation}

\title[Quasispecies Models with Delayed Replicase Production and Complementation]%{Quasispecies Models with Polymerase-Mediated Replication: Delay-Induced Extinction via Basin Reorganization}
{Polymerase-mediated quasispecies dynamics: bifurcations, complementation, and robustness of {\Large $\tau$}-tipping}
\author[E.A. Turner, F. Crespo, S.F. Elena, N. Morales, J. Sardanyés]{{Edward A. Turner$^{1*}$\orcidlink{0000-0002-2959-9227}, Francisco Crespo$^{2}$\orcidlink{0000-0002-5930-8523},   Nolbert Morales$^{5}$\orcidlink{0009-0009-8114-1537}, \\ 
Santiago F. Elena$^{3,4,7}$\orcidlink{000-0001-8249-5593},
%$^{6,7,8,9}$\orcidlink{0000-0003-4395-9708} 
\MakeLowercase{and} Josep Sardanyés$^{6,7}$\orcidlink{0000-0001-7225-5158}}}

\thanks{$^{*}$Corresponding author: \texttt{edward.turner@uvm.cl}}

\address{\footnotesize
$^{1}$Facultad de Ciencias Jurídicas, Sociales y de la Educación, Universidad Viña del Mar, Viña del Mar, Chile\\
$^{2}$Embry-Riddle Aeronautical University, Aerospace Engineering, Daytona Beach, FL 32114, USA\\
$^{3}$Institute for Integrative Systems Biology (I2SysBio), CSIC-Universitat de València,\\ Av. Catedrático Agustín E. 9, Paterna, 46980 València, Spain\\
$^{4}$The Santa Fe Institute, 1399 Hyde Park Road, Santa Fe, New Mexico 87501, USA\\
$^{5}$Facultad de Ingeniería, Universidad San Sebastián, Lago Panguipulli 1390, Puerto Montt, Chile\\
$^{6}$Centre de Recerca Matemàtica (CRM), Edifici C, Campus de Bellaterra,\\ 08193 Cerdanyola del Vallès, Barcelona, Spain\\
$^{7}$Dynamical Systems and Computational Virology, CSIC Associated Unit I2SysBio-CRM, Edifici C, Campus de Bellaterra, 08193 Cerdanyola del Vallès, Barcelona, Spain\\
}

	

%\linenumbers
\maketitle

\begin{abstract}
Quasispecies theory describes how mutation and selection shape highly variable RNA virus populations, but most models explicitly consider neither RNA-dependent RNA polymerases nor the delays associated with their synthesis and functional activation. A recent minimal model showed that delayed polymerase availability can induce extinction by reorganizing basins of attraction, a mechanism termed $\tau$-tipping. Here, we extend this framework to a delay differential equation model in which master and mutant genomes encode distinct functional polymerases. We characterize its equilibrium and bifurcation structure, identifying master-mutant coexistence, mutant-only persistence (error catastrophe), and complete extinction governed by transcritical and saddle-node bifurcations. Although delays leave the equilibria unchanged, they reorganize their basins of attraction and can redirect populations that would otherwise persist at fixed mutation and replication parameters toward extinction, thereby extending $\tau$-tipping to systems with autonomous mutant replication. We also recover the previously studied complementation model as a limiting case in which defective genomes depend on master-derived polymerase, providing a comprehensive analysis of the bifurcation structure. Comparing the two models shows that the ability of mutant genomes to encode functional polymerases determines whether loss of the master sequence results in mutant replacement or complete population extinction. The occurrence of delay-induced extinction in both settings demonstrates the robustness of $\tau$-tipping and connects intracellular replication kinetics, complementation, and quasispecies extinction. Finally, we translate this mechanism into concrete virological predictions and propose experimental strategies to test the effects of replicase timing, RNA degradation, and functional complementation on viral persistence.
\end{abstract}

%\newpage
%\tableofcontents

\section{Introduction}
RNA virus evolution is usually described in terms of mutation, selection, and competition among genome variants \cite{Eigen1988,Domingo1985,Sobrino1983,Ojosnegros2013,Mas2004,Sole2006}. Yet viral genomes do not replicate autonomously: their amplification depends on viral and host proteins whose production, processing, localization, and assembly introduce finite intracellular timescales. Consequently, two viral populations with the same genome composition and replication parameters may differ in their ability to establish infection if functional replication complexes become available at different times. How such intracellular timing affects the persistence and extinction of a viral quasispecies remains largely unresolved. RNA viruses exhibit exceptionally high mutation rates, primarily due to the limited copying fidelity of their RNA-dependent RNA polymerases (RdRp)~\cite{Drake1999,Sanjuan2010}. Consequently, rather than consisting of identical genomes, viral populations are composed of closely related variants that collectively evolve as a complex genetic entity known as a quasispecies \cite{Eigen1988,Domingo1985,Sobrino1983,Ojosnegros2013,Mas2004,Sole2006}. The quasispecies framework, introduced by Eigen and Schuster \cite{Eigen1971,Eigen1988}, provides a theoretical approach to understand the evolutionary dynamics of these mutationally rich populations, highlighting the interplay between mutation and selection. Originally conceptualized for describing the evolution of biological information in prebiotic evolution, early quasispecies models assumed deterministic dynamics of sequences and competition, thereby introducing fitness-dependent selection into population genetics. Later on, quasispecies theory has found extensive applications in other systems of replicons undergoing large mutation rates such as RNA viruses~\cite{Domingo1978,Domingo1985,Viralcuasi}, bacteria and prions~\cite{Ojosnegros2013}, or cancer cell populations~\cite{Sole2003,Sole2004,Brumera2006}. 

Over the past decades, a multitude of quasispecies models have been developed to address a variety of key biological processes relevant to virology, especially focusing on their impact on quasispecies diversity and in both information and extinction thresholds. For instance, stochasticity, spatial dynamics, different modes of replication, and epistasis, to cite a few (see~\cite{Sole2021,Sardanyes2024} for recent reviews). Despite these advances, certain aspects remain insufficiently explored, such as the explicit modeling of the production and maturation of viral proteins responsible for genome amplification. While some models have included both viral structural and non-structural proteins as state variables~\cite{Dahari2006,Sardanyes2009,Llopis2025,zitzmann2023}, the explicit consideration of time lags in the synthesis and activation of these proteins has been largely neglected in the framework of quasispecies~\cite{TurnerPRE}. Hence, the impact of lagged intracellular processing of viral proteins, particularly on critical evolutionary phenomena such as the error threshold or lethal mutagenesis, remains poorly studied.

A hallmark prediction of quasispecies theory is the so-called error threshold, a critical mutation rate beyond which the master (wildtype, wt) sequence loses viability and the population becomes dominated by the mutant specturm~\cite{Eigen1971,Eigen1988,Biebricher2005,Bull2005}. Furthermore, the high mutation rates characteristic of RNA viruses promote the frequent generation of defective viral genomes (DVGs), i.e., truncated or extensively mutated genomes that cannot replicate or establish infection autonomously and instead rely on functional proteins supplied in trans by co-infecting replication-competent viruses through complementation~\cite{DVG}. When packaged into viral particles, these defective genomes form defective interfering particles, which can alter infection dynamics by competing with viable viruses for replication resources, thereby affecting viral fitness and pathogenicity~\cite{DVG}.

Despite the analytical power of quasispecies theory, it often overlooks fundamental temporal features of the viral replication cycle. In particular, the intrinsic time lags arising from processes such as replicase enzyme synthesis, maturation, and genome replication are seldom modeled, even though they are known to shape the dynamics of infection. The interval times between genome availability and the onset of productive RNA replication is experimentally measurable and can range from minutes to hours, depending on the virus and experimental system~\cite{Barton1993,Binder2013}. These temporal lags are especially prominent in positive-sense single-stranded RNA viruses. Functional RNA replication commonly requires more than synthesis of the catalytic polymerase subunit. Depending on the virus, productive replication may require polyprotein cleavage, interaction with viral or host cofactors, membrane recruitment, and assembly of replication organelles. These processes create a finite interval between expression of replication proteins and productive RNA synthesis. In coronaviruses, for example, nsp12 functions within a multicomponent replication-transcription complex associated with virus-induced membranes, illustrating how polymerase availability can be temporally separated from productive genome amplification. SARS-CoV-2 quasispecies~\cite{JARY20201560,Sun2021}, where the RdRp (nsp12) is initially translated as part of a polyprotein (pp1ab), requiring subsequent proteolytic cleavage by viral proteases (PLpro and 3CLpro), and assembly with essential cofactors (nsp7 and nsp8) to form the active replication-transcription complex (RTC) within virus-induced double-membrane vesicles. Each stage, from folding and cofactor binding to membrane recruitment and conformational activation, introduces a distinct delay between RdRp synthesis and its first engagement in RNA replication~\cite{Kovski2021}. While such time lags are not always directly quantified in experimental studies, their existence and functional importance are increasingly recognized in molecular virology and structural biology~\cite{brown2013,smith2011}.

From a modeling perspective, the inclusion of time lags to represent such intracellular lags is both biologically motivated and mathematically necessary for realistic descriptions of viral dynamics. Classical models have employed delay differential equations (DDEs) to capture lags in intracellular processes, including the time between viral entry and virion production, or between gene expression and protein functionality~\cite{Sazonov2020}. For example, Mittler et al. (1998) analyzed a model for the interaction of human immunodeficiency virus type 1 (HIV-1) with target cells including a time lag between initial infection and the formation of productively infected cells ~\cite{Mittler1998}. Similar approaches have been used to study delayed gene expression in bacteriophage systems and lytic virus cycles~\cite{Culshaw2003}. Recently, time lags and periodic fluctuations were included in quasispecies models only considering an  abstract representation of viral genomes~\cite{Turner2025,Turner2026Quasi,QuasiGeneral2026}. However, the explicit modeling of time lags in the synthesis and functional maturation of viral polymerases remains largely absent from the quasispecies context. This gap leaves open the question of how lagged RdRp kinetics may influence quasispecies evolutionary features.

Recently, Turner et al.~\cite{TurnerPRE} introduced a minimal three-variable quasispecies model in which master genomes produce a shared RdRp, while defective mutant genomes relied on this polymerase through functional complementation. They showed that increasing the delay before the polymerase becomes functional can redirect trajectories from persistence to extinction by reorganizing the basins of attraction, without altering the equilibria or their local stability. This delay-induced transition, whose underlying dynamical mechanism was further characterized using a reduced two-variable model, was termed \(\tau\)-tipping~\cite{TurnerPRE}. However, the reduced model did not describe mutant genomes encoding their own functional polymerases and therefore could not address how polymerase autonomy, differential polymerase activity, and genome-polymerase preference modify the equilibrium and bifurcation structure.

Here, we extend that framework by introducing and analyzing a four-variable quasispecies model in which the master and mutant genome classes produce distinct functional RdRps. The model incorporates mutation, differential replication, genome-polymerase template preference, molecular degradation, and delays associated with RdRp synthesis, maturation, and activation. We first derive the equilibrium and bifurcation structure of the non-delayed system and determine the conditions for master-mutant coexistence, mutant-only persistence, and complete extinction. We then investigate whether the \(\tau\)-tipping mechanism identified in the reduced model [35] persists when both genome classes are capable of producing functional polymerases. Finally, we recover the complementation model as a limiting case describing defective genomes lacking functional polymerase genes and compare the bifurcation structures of the two models. This comparison allows us to determine how the functional autonomy of the mutant spectrum controls whether master-sequence loss produces mutant replacement or extinction of the complete viral population.

\section{Mathematical models and methods}
In this section, we introduce a mathematical model given by four coupled delay differential equations (DDEs). This model is based on the classical quasispecies framework originally proposed by Eigen and Schuster~\cite{Eigen1971,Eigen1988}, considering a master sequence (wildtype, wt) and the pool of mutants represented by an average mutant sequence (i.e., single-peak fitness landscape~\cite{Swetina1982,Sole2006}). Our approach, as a difference from the original model, assumes a non-constant population of genomes and incorporates polymerase-mediated replication, degradation rates of the genomes and the RdRps, along with time lags in the synthesis of functional RdRps.
The synthesis of both the genomes and the polymerases is assumed to be bounded and slowed down by logistic functions with carrying capacities equal to one. These functions introduce intra- and inter-specific competition between the viral genomes and proteins due to a finite availability of e.g., nucleotides and aminoacids inside the host cell, respectively.
A similar model without time lags in RdRp synthesis and constant populations was considered in Ref.~\cite{Sardanyes2010}.
Our model, as in ~\cite{Sardanyes2010}, includes a parameter ($\gamma$) determining the relative template-recognition or replication-allocation for the RdRps. That is, for $\gamma = 0.5$ both the master and the mutant genomes have the same affinity for the RdRps. Parameter $\gamma$ phenomenologically summarizes a number of biological properties such as \emph{cis}-acting replication signals, RNA secondary structures, template accessibility, intracellular colocalization, or preferential recruitment into replication complexes.

\subsection{Full model} We first introduce the principal extension of the minimal framework analyzed in Ref.~\cite{TurnerPRE}. Whereas the model in Ref.~\cite{TurnerPRE} contains a single polymerase produced by the master genome and shared with polymerase-defective mutants, the present model allows both the master and mutant genome classes to encode their own functional polymerases. The state variables are defined as follows: the master sequence ($x_0$), an average mutant sequences ($x_1$), and the master ($p_0$) and mutant  ($p_1$) RdRps. The model is given by the next system of DDEs 
\begin{equation}\small
\begin{split}\label{eq:fullmodel}
\frac{dx_0(t)}{dt}&= (1 - \mu)\, \gamma  \, x_0(t) \Big[r_0 p_0(t-\tau_0) + r_1 p_1(t-\tau_1)\Big] \, \Big(1 - x_0(t) - x_1(t)\Big) - \varepsilon\,  x_0(t),\\
\frac{dx_1(t)}{dt}&= \left(\mu \, \gamma \, x_0(t) + (1-\gamma) x_1(t) \right) \Big[r_0 p_0(t-\tau_0) + r_1 p_1(t-\tau_1)\Big] \Big(1 - x_0(t) - x_1(t)\Big) - \varepsilon \, x_1(t),\\ 
\frac{dp_0(t)}{dt}&= k_0\, x_0(t) \Big( 1 - p_0(t)-p_1(t)\Big) - \varepsilon_p \, p_0(t), \\
\frac{dp_1(t)}{dt}&= k_1 \, x_1(t) \Big( 1 - p_0(t)-p_1(t)\Big) - \varepsilon_p \, p_1(t).
\end{split}
\end{equation}
The parameter $r_i$ (with $i=0,1$) denotes the replication rate of population $i$, while $\mu$ is the mutation rate from the master sequence to the mutant, with backward mutations neglected. The parameter $\gamma$ denotes the fraction of RdRps (wild-type and mutant) that bind and replicate the master genomes, whereas the remaining fraction, $1-\gamma$, replicates the mutant genomes. This provides an explicit affinity parameter between genomes and RdRps~\cite{Sardanyes2010}. Furthermore, $k_i$ is the RdRp translation rate, $\varepsilon_p$ the polymerase degradation rate, and $\varepsilon$ the genome degradation rate. The delays $\tau_0$ and $\tau_1$ represent the synthesis, maturation, and activation times required for RdRp molecules to become functional. That is, the interval between polymerase production and its participation in genome replication. We assume these processes occur on slower timescales than RNA synthesis. Importantly, incorporating explicit delays is not equivalent to simply reducing the replication rates. Whereas lower replication rates only decrease the instantaneous amplification efficiency, delays introduce non-locality and memory into the dynamics because replication at time $t$ depends on the system state at earlier times $t-\tau_i$. Consequently, the delayed model is infinite-dimensional, since its evolution depends on the history over $[t-\tau_i,t]$ and cannot, in general, be reduced to a finite-dimensional system.

Our model incorporates mutational effects both at the level of the genome replication rates, $r_0$ and $r_1$, and at the level of the affinities of the two genome classes for RdRp binding. From a molecular perspective, these effects capture distinct biological mechanisms. Mutations may alter the catalytic efficiency, processivity, or replication speed of the RdRps, thereby affecting the replication rates of the corresponding genomes. At the same time, mutations can modify cis-acting RNA elements, secondary structures, or recognition motifs involved in polymerase recruitment and binding to viral RNA templates, thus changing genome affinities for replication. When $\gamma=0.5$, replicases interact equally with master and mutant genomes, so mutational effects act only through replicase activity. For $\gamma>0.5$, the affinity of the master genomes for the RdRp is larger than for the mutants, whereas $\gamma<0.5$ favors mutant-genome replication. The model thus retains a simplified quasispecies structure in which selective differences arise from both fidelity of replication and efficiencies due to genome-polymerase affinities.

As mentioned, we consider $r_0>r_1$ and set $r_0=1$, so that the wild-type RdRp transcribes RNA at a higher rate than the mutant RdRp. Consequently, genomes encoding the wild-type polymerase possess higher effective fitness than mutant genomes, providing the analog of a single-peak selective hierarchy. For simplicity, we set $\varepsilon_p=\varepsilon$, assuming that viral genomes and RdRp molecules are degraded on comparable intracellular timescales. This reduction avoids introducing an additional independent decay parameter and allows us to focus on the effects of mutation, genome--polymerase affinity, and delayed RdRp functionality. This assumption does not qualitatively change the mechanisms under study, since both $\varepsilon$ and $\varepsilon_p$ play analogous roles as loss terms for genomes and polymerases, respectively. Unless otherwise specified, we also set $\tau_0=\tau_1=\tau$, corresponding to the case in which wild-type and mutant RdRps undergo similar synthesis, maturation, and activation delays.

\subsection{Model with defective viral genomes}
To compare the full model with the minimal \(\tau\)-tipping framework introduced in Ref.~\cite{TurnerPRE}, we consider the limiting case in which mutant genomes retain the cis-acting signals required for recognition and replication but cannot produce a functional polymerase. These defective genomes therefore depend on polymerase supplied in trans by replication-competent master genomes. Under these assumptions, the four-variable model reduces to the three-variable complementation model analyzed in Ref. [35], subject to the parameter simplifications specified below. This model is given by:
\begin{equation}\small
\begin{split}\label{eq:DIviruses}
 \frac{dx_0(t)}{dt}&= r_0p_0(t-\tau) (1 - \mu) \gamma x_0(t) \Big(1 - x_0(t) - x_1(t)\Big) - \varepsilon\,  x_0(t),\\
 \frac{dx_1(t)}{dt}&= r_0p_0(t-\tau) \Big[ \mu \,\gamma\, x_0(t) + (1-\gamma) \,x_1(t) \Big]  \Big(1 - x_0(t) - x_1(t)\Big) - \varepsilon x_1(t),\\
 \frac{dp_0(t)}{dt}&= k_0 \, x_0(t) \left[1 - p_0(t)\right] - \varepsilon_p \, p_0(t).
\end{split}
\end{equation}
In Section~\ref{sec:complementation} we analyze system~\eqref{eq:DIviruses} under the simplifying assumption $\varepsilon_p=\varepsilon$, implying comparable degradation timescales for viral genomes and polymerases. As mentioned, this choice reduces the number of independent parameters while preserving the relevant competition between replication, degradation, and delayed polymerase availability.

\subsection{Numerical methods}
The numerical integration of the time-delayed system is performed using the adaptive Runge--Kutta method of Dormand--Prince (order 4/5) with automatic step-size control.

\section{Results and discussion}
\subsection{Dynamics of the full model without time lags}\label{sec:3.1}
We first consider the model~\eqref{eq:fullmodel} setting $\tau = 0$. We study the equilibrium points and their stability and also describe the invariant regions in the phase space. We also prove no periodic orbits exist.
\subsubsection{Domain and equilibrium points}\label{sec:3.1.1}
In many biological models, population growth is often constrained by a competition term for limited resources, restricting the growth and confining the system's dynamics within a bounded domain. In model \eqref{eq:fullmodel}, this limited region is represented by
\begin{equation}\nonumber
\mathcal D=\left\lbrace \mathbf x\in \mathbb R^4:x_0,x_1,p_0,p_1\geq0\, \wedge\, x_0+x_1\leq1\,\wedge\,p_0+p_1\leq1\right\rbrace,
\end{equation}
where $\mathbf x=(x_0,x_1,p_0,p_1).$ The coordinate space $x_0 = 0$ is invariant under the flow, and since the vector field $\eqref{eq:fullmodel}$ is negative on each boundary of $\mathcal{D}$, the manifold $\mathcal{D}$ remains invariant. Consequently, for any choice of parameters, the system stays within the biologically relevant region for all $t \geq 0$. 

\begin{Proposition}\label{prop:1R}
Let $\mathcal{A} = \left\{\mathbf{x} \in \mathcal{D} : \gamma\mu x_0 + \left(1 + \gamma(\mu - 2)\right) x_1 = 0 \right\}$. Then, the set $\mathcal{A}$ is invariant under the flow of system~\eqref{eq:fullmodel}. 
\end{Proposition}
\begin{proof}
A straightforward computation shows that the defining condition of the set $\mathcal{A}$ is preserved along the trajectories of system~\eqref{eq:fullmodel}.
\end{proof}

Consider the 6-dimensional parameter space
$$\Omega= \left\lbrace \omega\in\mathbb R^6: k_0,k_1,\gamma,\mu,r_1\in(0,1) \wedge 0<\varepsilon\ll1\right\rbrace$$
and the following functions: 
\begin{equation}
\begin{split}\label{eq:simpParmetros}
\alpha(\omega)&=k_0(1+\gamma(\mu-2)),\quad \ \beta(\omega)=k_1r_1\gamma\mu,\\
\kappa(\omega)&=k_1r_1(\gamma-1),\quad \ \ \qquad
\delta(\omega)=\gamma(\mu-1).
\end{split}
\end{equation}

For simplicity, let us denote $\alpha(\omega)$, $\beta(\omega)$, $\kappa(\omega)$, and $\delta(\omega)$ as $\alpha$, $\beta$, $\kappa$ and $\delta$.  The system $\eqref{eq:fullmodel}$ has five equilibria, denoted as $\mathbf x_i^*=(x_{0i},x_{1i},p_{0i},p_{1i})$. The firs equilibrium $\mathbf x^*_1 = (0, 0, 0, 0),$ represents the complete extinction of the system, i.e., lethality. The second and third equilibria, $\mathbf x_{2,3}^*,$ are given by
\begin{equation}
\begin{split}\nonumber
\mathbf x_{2,3}^*&=\left(0,\frac{\kappa+k_1\varepsilon \mp\sqrt{(\kappa +k_1 \varepsilon )^2+4 \kappa  \varepsilon ^2}}{2 \kappa },0,\frac{\kappa-k_1 \varepsilon  \mp\sqrt{(\kappa +k_1 \varepsilon )^2+4 \kappa  \varepsilon ^2}}{2 (\kappa +(\gamma -1) r_1\varepsilon )}\right).
\end{split}
\end{equation}
At these equilibria, the master sequence is completely lost, and the population is only composed of mutants, i.e., the error catastrophe scenario. The fourth and fifth equilibria, are given by  $\mathbf x_{4,5}^*=(x_{0(4,5)},x_{1(4,5)},p_{0(4,5)},p_{1(4,5)}),$ with
\begin{equation}\small
\begin{split}\nonumber
x_{0(4,5)}&=\frac{\alpha  \left[-\alpha(\delta+\varepsilon)+\beta(\delta+\varepsilon/r_1)\mp\sqrt{(\delta(\alpha-\beta)+\varepsilon  (\alpha -\beta/r_1 ))^2-4 \delta \varepsilon ^2 (2 \gamma -1)(\alpha -\beta ) }\right]}{k_0  \delta (4 \gamma -2)(\alpha -\beta )},\\
x_{1(4,5)}&= \frac{\mu  \left[\alpha  (\delta +\varepsilon )- \beta (\delta - \varepsilon/r_1)\pm\sqrt{(\delta(\alpha-\beta)+\varepsilon  (\alpha -\beta/r_1 ))^2-4 \delta \varepsilon ^2 (2 \gamma -1)(\alpha -\beta ) }\right]}{(4 \gamma -2) (\mu -1) (\alpha -\beta )},\\
p_{0(4,5)}&= \frac{\alpha  \left[\alpha  (\delta -\varepsilon )-\beta  \left(\delta-\varepsilon/ r_1 \right)\pm\sqrt{(\delta(\alpha-\beta)+\varepsilon  (\alpha -\beta/r_1 ))^2-4 \delta \varepsilon ^2 (2 \gamma -1)(\alpha -\beta ) }\right]}{2 \delta (\alpha +(1-2 \gamma ) \varepsilon -\gamma  k_1 \mu )(\alpha -\beta )},\\
p_{1(4,5)}&=\frac{k_1 \mu  \left[-\alpha  (\delta -\varepsilon )+\beta  \left(\delta-\varepsilon/r_1 \right)\mp\sqrt{(\delta(\alpha-\beta)+\varepsilon  (\alpha -\beta/r_1 ))^2-4 \delta \varepsilon ^2 (2 \gamma -1)(\alpha -\beta ) }\right]}{2 (\mu -1) (\alpha +(1-2 \gamma ) \varepsilon -\gamma  k_1 \mu )(\alpha -\beta )}.
\end{split}
\end{equation}
Note that the upper sign in the numerators corresponds to the equilibrium $\mathbf{x}_4^*$, while the lower sign corresponds to $\mathbf{x}_5^*$. Both equilibria represent coexistence between the master sequence and its mutants.
%Furthermore, we recall that the equilibrium points $\mathbf{x}_1^*$, $\mathbf{x}_2^*$, and $\mathbf{x}_3^*$ lie on the invariant set $\mathcal{A}$.
In what follows, Proposition~\ref{prop:x1yx2} and Proposition~\ref{eq:eq4y5} establish that, for certain parameter values $\omega$, the equilibria lie within the biologically relevant region $\mathcal{D}$.

\begin{Proposition}\label{prop:x1yx2}
The equilibria $\mathbf{x}_2^*$ and $\mathbf{x}_3^*$ lie inside the bounded region $\mathcal{D}$ 
if one of the following cases holds:
\begin{enumerate}
\item $4r_1(1-\gamma)<1$, $k_1>4r_1(1-\gamma)$, and
$\varepsilon \leq \dfrac{\sqrt{k_1}\, r_1(1-\gamma)}{\sqrt{k_1}+2\sqrt{r_1(1-\gamma)}}.$
\item $4r_1(1-\gamma)>1$, $k_1<1$, and $\varepsilon \leq \dfrac{\sqrt{k_1}\, r_1(1-\gamma)}{\sqrt{k_1}+2\sqrt{r_1(1-\gamma)}}.$
\item $4r_1(1-\gamma)=k_1<1$ and $\varepsilon \leq \dfrac{r_1(1-\gamma)}{2}.$
\end{enumerate}
\end{Proposition}

\begin{figure}
\centerline{\includegraphics[width=\textwidth]{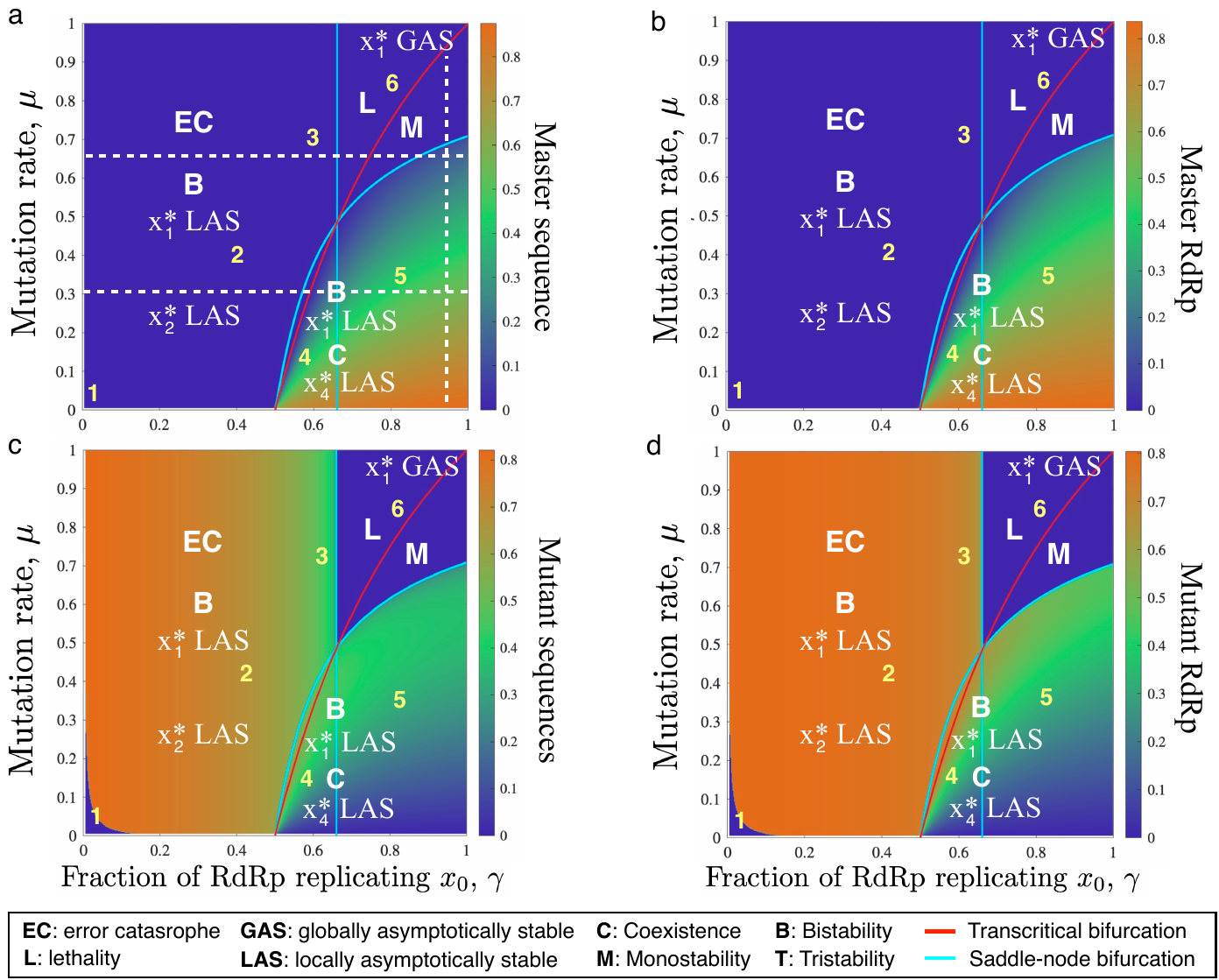}}
  \captionsetup{width=\linewidth}
\caption{Phase diagram in the $(\gamma,\mu)$ parameter space, computed numerically from Eqs.~\eqref{eq:fullmodel} for $\tau=0$. The saddle-node and transcritical bifurcation boundaries are overlaid in cyan and red, respectively. Distinct dynamical regimes corresponding to error catastrophe (EC) and lethality (L) are identified, together with the stable equilibrium states found in each region. In the bistable region (B), the equilibria $x^*_{1}$, $x^*_{2}$, and $x^*_{1}$, $x^*_{4}$ are locally asymptotically stable (LAS). Parameter values are $r_0=1$, $r_1=0.7$, $k_0=0.6$, $k_1=0.5$, and $\varepsilon=\varepsilon_p=0.1$; and initial conditions are \textcolor{black}{$x_0(0) = 0.2$, $x_1(0) = 0$, $p_0(0) = 0$, $p_1(0) = 0$}. The white dashed lines indicate parameter-space sections used to construct the bifurcation diagrams shown in the following figure. The numbers highlighted in yellow identify the different regions of the phase diagram for which the equilibrium points and their corresponding eigenvalues are characterized in Table~\ref{tab:eigenvalues}.}
\label{fig:Gradiente4D}		
\end{figure}

\begin{Proposition}\label{eq:eq4y5}
The equilibria $\mathbf{x}_4^*$ and $\mathbf{x}_5^*$ lie in the bounded region $\mathcal{D}$ provided that
$$k_1 r_1 \leq k_0, \qquad \mu \leq 2 - \frac{1}{\gamma},$$
and
$$\varepsilon \leq \dfrac{\gamma(1-\mu)\big(k_1\gamma\mu r_1-k_0(1-2\gamma+\gamma\mu)\big)}{k_1\gamma\mu -k_0(1-2\gamma+\gamma\mu)+2\sqrt{(1-2\gamma)\gamma(1-\mu)\big(k_0(1-2\gamma+\gamma\mu)-k_1 r_1 \gamma \mu\big)}}.$$
\end{Proposition}

\begin{figure}
\centerline{\includegraphics[width=\textwidth]{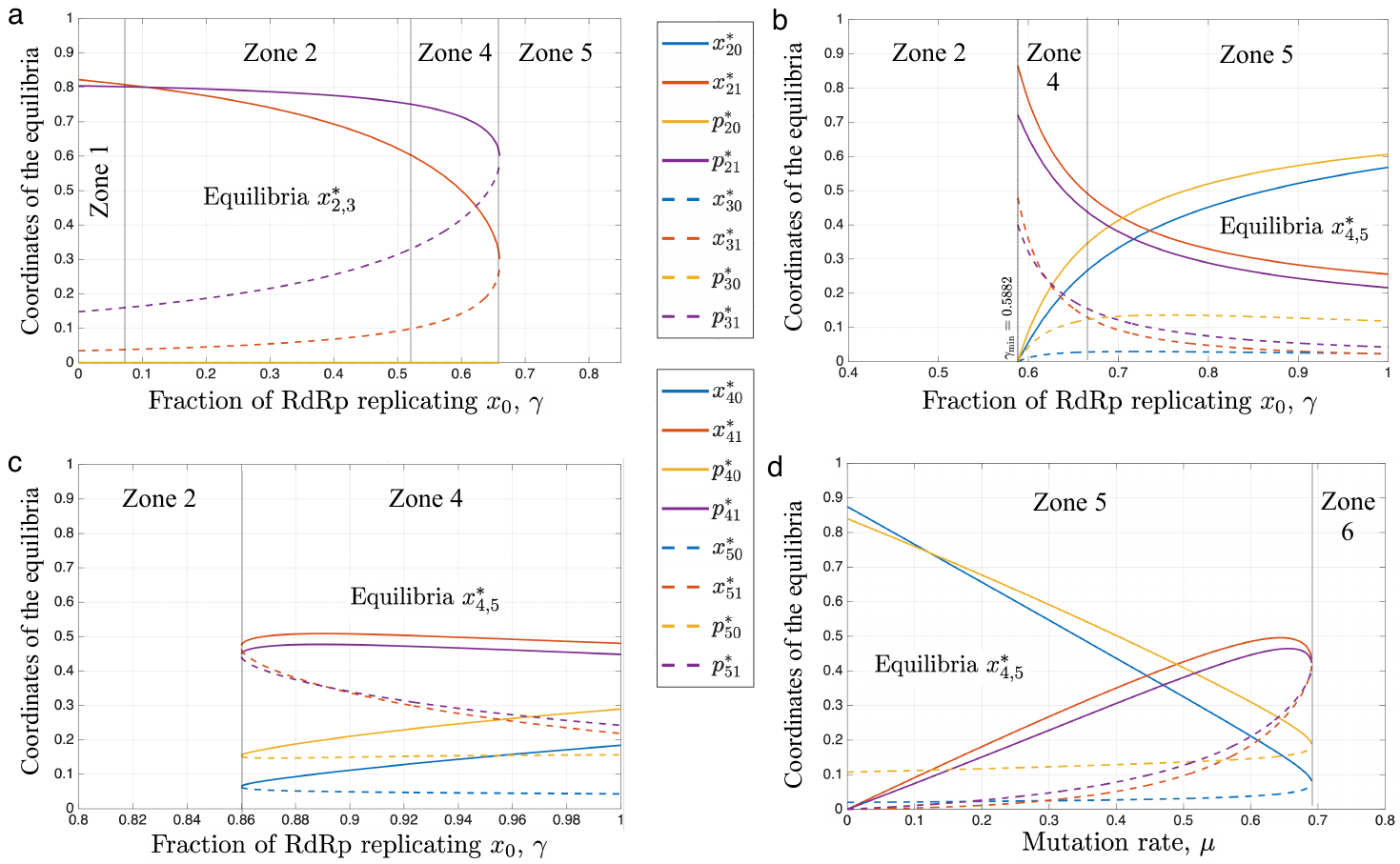}}
  \captionsetup{width=\linewidth}
\caption{Bifurcation diagrams for the equilibria $x^*_{2,3}$ and $x^*_{4,5}$. The parameter values and initial conditions are set as in the previous figure. Solid and dashed curves denote the coordinates of the two equilibrium branches indicated in each panel. Shaded regions correspond to the parameter-space zones identified in Fig.~\ref{fig:Gradiente4D}.
(a) Coordinates of the mutant-only equilibria $x_2^*$ and $x_3^*$ as functions of the genome–polymerase affinity parameter $\gamma$, for $\mu=0.01$, showing a saddle-node bifurcation between Zones 4 and 5.
(b) Coordinates of the coexistence equilibria $x_4^*$ and $x_5^*$ as functions of $\gamma$, for $\mu=0.3$, showing the transcritical transition into Zone 4, where master–mutant coexistence becomes biologically admissible.
(c) Coordinates of the coexistence equilibria $x_4^*$ and $x_5^*$ as functions of $\gamma$, for $\mu=0.65$, showing a saddle-node bifurcation from Zone 6 to Zone 5.
(d) Coordinates of $x_4^*$ and $x_5^*$ as functions of the mutation rate $\mu$, with $\gamma=0.95$ fixed, showing a saddle-node bifurcation separating Zone 5 from Zone 6.}
\label{fig:1D_bif}		
\end{figure}

%\edcom[inline]{observemos que el equilibrio $\mathbf{x}_4^*,\mathbf{x}_5^*$ se indetermina en 0.5 pero nunca está en T cuando ocurre.}

The previous propositions identify a critical threshold relationship between $\mu$ and $\gamma$ governing the persistence of the master genome. Such a dependence was already observed in the quasispecies framework by Sardanyés and Elena~\cite{Sardanyes2010}. Here, however, the inclusion of genetic degradation, quantified by $\varepsilon$, sharpens this condition by further restricting the parameter combinations for which the persistence equilibria are biologically admissible. In biological terms, higher mutation rates, weaker master genome--polymerase affinity, or stronger genome degradation all reduce the parameter region where coexistence remains feasible, thereby increasing the likelihood of error catastrophe and  viral extinction, i.e., lethality. The resulting feasible region is shown in Fig.~\ref{fig:Gradiente4D} (see next section).

\subsubsection{Local stability analysis of the equilibrium points}\label{sec:3.1.2}
In this section, we study analytically and numerically the local stability of the equilibria previously described whenever they lie within $\mathcal{D}$. The stability of the fixed point $x_1^*$ is determined from the linearized system
\[
\frac{d}{dt}\bigl(\mathbf{x}-\mathbf{x}_1^*\bigr)
=
DF(\mathbf{x}_1^*)\bigl(\mathbf{x}-\mathbf{x}_1^*\bigr).
\]

At this equilibrium, the Jacobian matrix is
\[
DF(\mathbf{x}_1^*)=
\begin{pmatrix}
-\varepsilon & 0 & 0 & 0\\
0 & -\varepsilon & 0 & 0\\
k_0 & 0 & -\varepsilon & 0\\
0 & k_1 & 0 & -\varepsilon
\end{pmatrix}.
\]
Since this matrix is lower triangular, its eigenvalues are given by its diagonal entries:
\[
\lambda_{11}=\lambda_{12}=\lambda_{13}=\lambda_{14}
=-\varepsilon<0.
\]
Therefore, $\mathbf{x}_1^*$ is locally asymptotically stable since $\varepsilon >0$ and, consequently, it is a sink independently of the choice of $\omega\in\Omega$.

The stability of the other four equilibria is computed numerically in the parameter space $(\gamma, \mu) \in [0,1]$. The phase diagrams in Fig.~\ref{fig:Gradiente4D} reveal a rich organization of the dynamics in this parameter space. At low values of $\gamma$, replication is preferentially directed toward mutant genomes and the master sequence cannot be maintained, giving rise to the error-catastrophe (EC) regime. Increasing $\gamma$ favors replication of the master sequence and allows the emergence of a coexistence equilibrium in which master and mutant genomes, together with their corresponding RdRps, persist. However, this transition is strongly modulated by the mutation rate: increasing $\mu$ progressively restricts the parameter region supporting coexistence and eventually drives the system toward either lower equilibria of master genomes or to complete extinction. Importantly, the phase diagram contains extended bistable zones (denoted as B), indicating that identical parameter values can lead to qualitatively different long-term outcomes depending on the initial population composition. The boundaries separating these regimes are organized by transcritical (red curves) and saddle-node (cyan curves) bifurcations, which respectively mediate exchanges between biologically relevant equilibrium branches and the creation or loss of equilibrium pairs. Thus, mutation and genome--polymerase affinity jointly determine not only which viral states are feasible, but also whether persistence or extinction is uniquely determined or depends on initial conditions. The full-extinction state, which is monostable (M), is found at large values of both $\gamma$ and $\mu$. This threshold can be crossed by increasing the mutation rate when RdRp affinity is strongly biased toward the master genome ($\gamma \to 1$). Interestingly, this result indicates that the selective advantage conferred to the master sequence by preferential access to the replication machinery may become detrimental at sufficiently high mutation rates. Under these conditions, enhanced replication of the master sequence also increases the production of mutants, ultimately promoting the collapse of the entire viral population.

This bifurcation structure is resolved in greater detail in Fig.~\ref{fig:1D_bif} and is consistent with the local stability analysis summarized in Table~\ref{tab:eigenvalues}. The one-parameter diagrams show how the equilibrium branches underlying the different regions of Fig.~\ref{fig:1D_bif} are connected. For low mutation ($\mu=0.01$), variation of $\gamma$ produces a saddle-node bifurcation of the mutant-only equilibria $\mathbf{x}^{*}_{2,3}$, whereas at intermediate mutation ($\mu=0.3$) the coexistence equilibria $\mathbf{x}^{*}_{4,5}$ become biologically admissible through a transcritical transition. At larger mutation ($\mu=0.65$), coexistence is instead bounded by a saddle-node bifurcation, a behavior that is also recovered when $\mu$ is varied at fixed high affinity of the master genomes ($\gamma=0.95$). The eigenvalues reported in Table~\ref{tab:eigenvalues} confirm the corresponding changes in local stability across these regions. In particular, the extinction equilibrium $\mathbf{x}^{*}_{1}$ remains locally asymptotically stable throughout parameter space, with all four eigenvalues equal to $-\varepsilon$, while the stability and biological admissibility of the mutant-only and coexistence equilibria change across the bifurcation boundaries. Taken together, Figs.~\ref{fig:Gradiente4D}--\ref{fig:1D_bif} and Table~\ref{tab:eigenvalues} show that the non-delayed model possesses a multistable landscape in which error catastrophe, coexistence, and lethality are connected through a well-defined sequence of local bifurcations. This underlying equilibrium structure provides the reference dynamical landscape for assessing how the introduction of delays in RdRp production and maturation subsequently reorganizes the basins of attraction without changing the equilibrium set (see Section~\ref{sec:3.2}).

This bifurcation structure extends the classical single-peak quasispecies picture in two ways. First, replication success depends not only on genome-level replication and mutation parameters but also on the production of a \emph{trans}-acting replicase. Second, extinction is locally stable over the entire parameter domain, generating bistability between extinction and positive viral states. Consequently, the mutation-affinity boundary does not by itself determine infection outcome. On one side of the classical threshold, mutant genomes may replace the master while preserving total viral replication; elsewhere, the same mutation pressure can contribute to population collapse. This distinction parallels the virological difference between loss of a dominant sequence and extinction of the complete mutant spectrum.

\begin{table}
\centering
\renewcommand{\arraystretch}{1.25}
\setlength{\tabcolsep}{4.5pt}

\resizebox{\textwidth}{!}{%
\begin{tabular}{|>{\centering\arraybackslash}m{0.8cm}|*{6}{>{\centering\arraybackslash}m{0.9cm}|>{\centering\arraybackslash}m{2.2cm}|}}
\hline
&
\multicolumn{2}{c|}{\textbf{\makecell{Zone 1\\ $\gamma=0.01,\ \mu=0.01$}}}
& \multicolumn{2}{c|}{\textbf{\makecell{Zone 2\\ $\gamma=0.4,\ \mu=0.4$}}}
& \multicolumn{2}{c|}{\textbf{\makecell{Zone 3\\ $\gamma=0.6,\ \mu=0.7$}}}
& \multicolumn{2}{c|}{\textbf{\makecell{Zone 4\\ $\gamma=0.6,\ \mu=0.1$}}}
& \multicolumn{2}{c|}{\textbf{\makecell{Zone 5\\ $\gamma=0.8,\ \mu=0.3$}}}
& \multicolumn{2}{c|}{\textbf{\makecell{Zone 6\\ $\gamma=0.95,\ \mu=0.8$}}} \\
\hline

\textbf{$x_1^*$}
& \makecell{$\in\mathcal{D}$\\ $\in\mathbb{R}^+$}
& \makecell{$\lambda_1=-0.1$\\ $\lambda_2=-0.1$\\ $\lambda_3=-0.1$\\ $\lambda_4=-0.1$}
& \makecell{$\in\mathcal{D}$\\ $\in\mathbb{R}^+$}
& \makecell{$\lambda_1=-0.1$\\ $\lambda_2=-0.1$\\ $\lambda_3=-0.1$\\ $\lambda_4=-0.1$}
& \makecell{$\in\mathcal{D}$\\ $\in\mathbb{R}^+$}
& \makecell{$\lambda_1=-0.1$\\ $\lambda_2=-0.1$\\ $\lambda_3=-0.1$\\ $\lambda_4=-0.1$}
& \makecell{$\in\mathcal{D}$\\ $\in\mathbb{R}^+$}
& \makecell{$\lambda_1=-0.1$\\ $\lambda_2=-0.1$\\ $\lambda_3=-0.1$\\ $\lambda_4=-0.1$}
& \makecell{$\in\mathcal{D}$\\ $\in\mathbb{R}^+$}
& \makecell{$\lambda_1=-0.1$\\ $\lambda_2=-0.1$\\ $\lambda_3=-0.1$\\ $\lambda_4=-0.1$}
& \makecell{$\in\mathcal{D}$\\ $\in\mathbb{R}^+$}
& \makecell{$\lambda_1=-0.1$\\ $\lambda_2=-0.1$\\ $\lambda_3=-0.1$\\ $\lambda_4=-0.1$} \\
\hline

\rowcolor{gray!10}
\textbf{$x_2^*$}
& \makecell{$\in\mathcal{D}$\\ $\in\mathbb{R}^+$}
& \makecell{$\lambda_1=-0.587$\\ $\lambda_2=-0.38$\\ $\lambda_3=-0.1$\\ $\lambda_4=-0.099$}
& \makecell{$\in\mathcal{D}$\\ $\in\mathbb{R}^+$}
& \makecell{$\lambda_1=-0.485$\\ $\lambda_2=-0.187$\\ $\lambda_3=-0.1$\\ $\lambda_4=-0.06$}
& \makecell{$\in\mathcal{D}$\\ $\in\mathbb{R}^+$}
& \makecell{$\lambda_1=-0.385$\\ $\lambda_2=-0.1$\\ $\lambda_3=-0.064$\\ $\lambda_4=-0.055$}
& \makecell{$\in\mathcal{D}$\\ $\in\mathbb{R}^+$}
& \makecell{$\lambda_1=-0.385$\\ $\lambda_2=-0.1$\\ $\lambda_3=-0.064$\\ \textcolor{blue}{$\lambda_4=0.035$}}
& \makecell{$\notin\mathcal{D}$\\ $\in\mathbb{C}$}
& \makecell{--\\--\\--\\--}
& \makecell{$\notin\mathcal{D}$\\ $\in\mathbb{C}$}
& \makecell{--\\--\\--\\--} \\
\hline

\textbf{$x_3^*$}
& \makecell{$\in\mathcal{D}$\\ $\in\mathbb{R}^+$}
& \makecell{$\lambda_1=-0.175$\\ $\lambda_2=-0.1$\\ $\lambda_3=-0.099$\\ \textcolor{blue}{$\lambda_4=0.054$}}
& \makecell{$\in\mathcal{D}$\\ $\in\mathbb{R}^+$}
& \makecell{$\lambda_1=-0.189$\\ $\lambda_2=-0.1$\\ $\lambda_3=-0.06$\\ \textcolor{blue}{$\lambda_4=0.047$}}
& \makecell{$\in\mathcal{D}$\\ $\in\mathbb{R}^+$}
& \makecell{$\lambda_1=-0.220$\\ $\lambda_2=-0.1$\\ $\lambda_3=-0.055$\\ \textcolor{blue}{$\lambda_4=0.032$}}
& \makecell{$\in\mathcal{D}$\\ $\in\mathbb{R}^+$}
& \makecell{$\lambda_1=-0.220$\\ $\lambda_2=-0.1$\\ \textcolor{blue}{$\lambda_3=0.035$}\\ \textcolor{blue}{$\lambda_4=0.032$}}
& \makecell{$\notin\mathcal{D}$\\ $\in\mathbb{C}$}
& \makecell{--\\--\\--\\--}
& \makecell{$\notin\mathcal{D}$\\ $\in\mathbb{C}$}
& \makecell{--\\--\\--\\--} \\
\hline

\rowcolor{gray!10}
\textbf{$x_4^*$}
& \makecell{$\notin\mathcal{D}$\\ $\in\mathbb{R}^-$}
& \makecell{--\\--\\--\\--}
& \makecell{$\notin\mathcal{D}$\\ $\in\mathbb{R}^-$}
& \makecell{--\\--\\--\\--}
& \makecell{$\notin\mathcal{D}$\\ $\in\mathbb{R}^-$}
& \makecell{\textcolor{blue}{$\lambda_1=0.184$}\\ \textcolor{blue}{$\lambda_2=0.122$}\\ $\lambda_3=-0.1$\\ $\lambda_4=-0.032$}
& \makecell{$\in\mathcal{D}$\\ $\in\mathbb{R}^+$}
& \makecell{$\lambda_1=-0.567$\\ $\lambda_2=-0.265$\\ $\lambda_3=-0.1$\\ $\lambda_4=-0.026$}
& \makecell{$\in\mathcal{D}$\\ $\in\mathbb{R}^+$}
& \makecell{$\lambda_1=-0.560$\\ $\lambda_2=-0.265$\\ $\lambda_3=-0.1$\\ $\lambda_4=-0.064$}
& \makecell{$\notin\mathcal{D}$\\ $\in\mathbb{C}$}
& \makecell{--\\--\\--\\--} \\
\hline

\textbf{$x_5^*$}
& \makecell{$\notin\mathcal{D}$\\ $\in\mathbb{R}^-$}
& \makecell{--\\--\\--\\--}
& \makecell{$\notin\mathcal{D}$\\ $\in\mathbb{R}^-$}
& \makecell{--\\--\\--\\--}
& \makecell{$\notin\mathcal{D}$\\ $\in\mathbb{R}^-$}
& \makecell{\textcolor{blue}{$\lambda_1=0.446$}\\ \textcolor{blue}{$\lambda_2=0.122$}\\ $\lambda_3=-0.1$\\ \textcolor{blue}{$\lambda_4=0.031$}}
& \makecell{$\in\mathcal{D}$\\ $\in\mathbb{R}^+$}
& \makecell{$\lambda_1=-0.183$\\ $\lambda_2=-0.1$\\ \textcolor{blue}{$\lambda_3=0.051$}\\ $\lambda_4=-0.026$}
& \makecell{$\in\mathcal{D}$\\ $\in\mathbb{R}^+$}
& \makecell{$\lambda_1=-0.183$\\ $\lambda_2=-0.1$\\ $\lambda_3=-0.064$\\ \textcolor{blue}{$\lambda_4=0.051$}}
& \makecell{$\notin\mathcal{D}$\\ $\in\mathbb{C}$}
& \makecell{--\\--\\--\\--} \\
\hline

\end{tabular}%
}
\captionsetup{width=\textwidth}
\caption{Eigenvalues associated with equilibria $x_i^*$, with  $i=1,\ldots,5$, for six representative zones of the parameter space $(\gamma,\mu)$, shown with yellow numbers in Fig.~\ref{fig:Gradiente4D}. For each equilibrium, the first entry indicates whether the equilibrium belongs to the biologically feasible region $\mathcal{D}$ and whether its coordinates are real and non-negative. The corresponding eigenvalues are listed in the adjacent column. Black values indicate eigenvalues with negative real part, whereas values in blue denote positive eigenvalues and therefore local instability.}
\label{tab:eigenvalues}
\end{table}

\subsubsection{Non-existence of periodic orbits}\label{sec:3.1.3}
Here, we show that the non-delayed system does not admit periodic solutions in the biologically feasible region $\mathcal{D}$. The proof is based on suitable ratios between the master and mutant populations and their corresponding polymerases. Periodicity forces these ratios to remain constant, thereby reducing any hypothetical periodic orbit to a planar subsystem. The Bendixson--Dulac criterion is then used to exclude periodic trajectories in both the interior and the relevant invariant boundary of $\mathcal{D}$.

\begin{Proposition}\label{prop:no-periodic-orbits}
Assume that $\tau_0=\tau_1=0$, $\varepsilon_p=\varepsilon$, and $0<\mu,\gamma<1$, $r_0,r_1,k_0,k_1,\varepsilon>0$. Then system~\eqref{eq:fullmodel} has no nonconstant periodic orbits in $\mathcal D$.
\end{Proposition}

\begin{proof}
Suppose, by contradiction, that $(x_0(t),x_1(t),p_0(t),p_1(t))$ is a nonconstant $T$-periodic solution contained in $\mathcal D$. We first consider the case $x_0\not\equiv0$. Since the hyperplane is invariant. Set
\begin{equation}\nonumber
X=x_0+x_1,\qquad P=p_0+p_1,\qquad R=r_0p_0+r_1p_1.
\end{equation}
On the boundary $X=1$, $\dot X=-\varepsilon<0,$ whereas on $P=1$, $\dot P=-\varepsilon<0.$ Since $X$ and $P$ are periodic and satisfy $X,P\leq1$, neither can attain the value $1$. Hence,
 $X(t)<1$ and $P(t)<1$ for all $t.$ We also have
$$\dot p_0=k_0x_0(1-P)-\varepsilon p_0.$$
If $p_0(t_0)=0$ for some $t_0$, then \(t_0\) is a minimum of the nonnegative periodic function $p_0$, but
$$\dot p_0(t_0)=k_0x_0(t_0)(1-P(t_0))>0,$$ a contradiction. Therefore, $p_0(t)>0$ for all $t$, and consequently \(R(t)>0\).
Consider
\[
q(t)=\frac{x_1(t)}{x_0(t)}.
\]
A direct computation gives $\dot q = R(1-X)\left[\gamma\mu+cq\right],$ with $c=1-2\gamma+\gamma\mu.$ If $c\geq0$, then $\dot q>0$, contradicting the periodicity of $q$. Thus, a periodic solution could only exist if $c<0$. Define
$$q_*=-\frac{\gamma\mu}{c}>0.$$
Then, $ d/dt(q-q_*)=cR(1-X)(q-q_*).$ Hence,
$$
q(T)-q_*=
\bigl(q(0)-q_*\bigr)
\exp\left(
c\int_0^T R(s)(1-X(s))\,ds
\right).
$$
The exponential factor is strictly smaller than one. Since $q(T)=q(0)$, it follows that $q(t)\equiv q_*,$ and therefore $x_1(t)=q_*x_0(t).$ Next, let
$$z(t)=\frac{p_1(t)}{p_0(t)}.$$
Using $x_1=q_*x_0$, we obtain
$$\dot z=
\frac{x_0(1-P)}{p_0}
\left(k_1q_*-k_0z\right).$$
Setting
$$z_*=\frac{k_1}{k_0}q_*,$$
we find
\[
\frac{d}{dt}(z-z_*)
=
-k_0\frac{x_0(1-P)}{p_0}(z-z_*).
\]
The coefficient is strictly negative. Periodicity therefore implies $z(t)\equiv z_*,$ and hence $p_1(t)=z_*p_0(t).$
Writing \(x=x_0\) and \(p=p_0\), the dynamics along the periodic solution reduces to
\begin{equation}\nonumber
\begin{aligned}
\dot x
&=
x\left[Ap(1-Bx)-\varepsilon\right],\\
\dot p
&=
k_0x(1-Cp)-\varepsilon p,
\end{aligned}
\end{equation}
where $A=(1-\mu)\gamma(r_0+r_1z_*)>0,$ $B=1+q_*>0$ and $C=1+z_*>0.$ The periodic solution is contained in
$\mathcal D=\left\{(x,p):x>0,\ p>0,\ Bx<1,\ Cp<1\right\}.$
This region is simply connected. For
\[
\mathcal B(x,p)=\frac1{xp},
\]
we obtain
\[
\frac{\partial}{\partial x}\bigl(\mathcal B \dot x\bigr)
+
\frac{\partial}{\partial p}\bigl(\mathcal B \dot p\bigr)
=
-AB-\frac{k_0}{p^2}<0
\]
throughout $\mathcal D$. The Bendixson--Dulac criterion excludes
nonconstant periodic solutions in $\mathcal D$, which is a
contradiction.
It remains to consider $x_0\equiv0$. In this case,
$\dot p_0=-\varepsilon p_0,$
and periodicity gives $p_0\equiv0$. The remaining subsystem is
\[
\begin{aligned}
\dot x_1
&=
x_1\left[(1-\gamma)r_1p_1(1-x_1)-\varepsilon\right],\\
\dot p_1
&=
k_1x_1(1-p_1)-\varepsilon p_1.
\end{aligned}
\]
If $x_1\equiv0$, then periodicity implies $p_1\equiv0$, and the
solution is an equilibrium. Otherwise, $x_1>0$. Moreover, the same
boundary argument gives $0<p_1(t)<1.$
For
\[
\widetilde{\mathcal B}(x_1,p_1)=\frac1{x_1p_1},
\]
we have
\[
\frac{\partial}{\partial x_1}
\left(\widetilde{\mathcal B}\dot x_1\right)
+
\frac{\partial}{\partial p_1}
\left(\widetilde{\mathcal B}\dot p_1\right)
=
-(1-\gamma)r_1-\frac{k_1}{p_1^2}<0.
\]
The Bendixson--Dulac criterion again excludes nonconstant periodic solutions. Therefore, system~\eqref{eq:fullmodel} admits no nonconstant periodic solutions contained in $\mathcal D$.
\end{proof}

\subsection{Dynamics of the full model with time lags}\label{sec:3.2}

Before analyzing the delayed dynamics, we first observe that the introduction of the delays $\tau_0$ and $\tau_1$ does not modify the equilibrium set of system~\eqref{eq:fullmodel}. Indeed, if $E^*=(x_0^*,x_1^*,p_0^*,p_1^*)$ is an equilibrium, then the corresponding solution is constant in time. Consequently, $p_0(t-\tau_0)=p_0^*$, $p_1(t-\tau_1)=p_1^*$ and the delays disappear from the equilibrium conditions. Therefore, the equilibria are determined by the algebraic system
\begin{equation}\nonumber
\begin{cases}
0=(1-\mu)\gamma x_0^*
\bigl(r_0p_0^*+r_1p_1^*\bigr)
\bigl(1-x_0^*-x_1^*\bigr)
-\varepsilon x_0^*,\\[1mm]
0=
\bigl[\mu\gamma x_0^*+(1-\gamma)x_1^*\bigr]
\bigl(r_0p_0^*+r_1p_1^*\bigr)
\bigl(1-x_0^*-x_1^*\bigr)
-\varepsilon x_1^*,\\[1mm]
0=k_0x_0^*
\bigl(1-p_0^*-p_1^*\bigr)
-\varepsilon_p p_0^*,\\[1mm]
0=k_1x_1^*
\bigl(1-p_0^*-p_1^*\bigr)
-\varepsilon_p p_1^*.
\end{cases}
\end{equation}
This system coincides with the equilibrium equations obtained from system~\eqref{eq:fullmodel} when $\tau_0=\tau_1=0$. Hence, the existence and location of the equilibrium points are independent of the delays, and the equilibria characterized in Subsection~\ref{sec:3.1.1} remain valid for arbitrary $\tau_0,\tau_1\geq0$.

Biologically, this means that the intracellular time required for polymerase synthesis, maturation, and activation does not alter the set of potential long-term population states. Instead, these kinetic processes influence how the system approaches those states and may therefore modify the transient dynamics and the stability properties of the equilibria. Consequently, any qualitative changes induced by the delays must arise through changes in the organization of the phase space rather than through the creation or disappearance of equilibrium points.

\begin{Proposition}\label{prop:1}
Let $\tau=\max\{\tau_0,\tau_1\},$ and define 
$\mathcal A=\left\{\mathbf x\in\mathcal D:\gamma\mu x_0+\bigl[1+\gamma(\mu-2)\bigr]x_1=0\right\}.$
If $\boldsymbol\phi\in C([-\tau,0],\mathcal A),$ then the corresponding solution of system~\eqref{eq:fullmodel} satisfies
$\mathbf x(t)\in\mathcal A$ for all $t\geq0.$
Consequently, $C([-\tau,0],\mathcal A)$ is positively invariant under the semiflow generated by system~\eqref{eq:fullmodel}.
\end{Proposition}

\begin{proof}
Let $\mathbf x(t)=\bigl(x_0(t),x_1(t),p_0(t),p_1(t)\bigr)$ be the solution corresponding to $\boldsymbol\phi\in C([-\tau,0],\mathcal A)$. Set $X(t)=x_0(t)+x_1(t),$ $P(t)=p_0(t)+p_1(t),$ and
$$R(t)=r_0p_0(t-\tau_0)+r_1p_1(t-\tau_1).$$ 
For every history segment taking values in $\mathcal D$, the delayed variables satisfy $p_0(t-\tau_0),p_1(t-\tau_1)\geq0$, and therefore $R(t)\geq0$.
Adding the first two equations of~\eqref{eq:fullmodel} gives
$$\dot X=\bigl[\gamma x_0+(1-\gamma)x_1\bigr]R(t)(1-X)-\varepsilon X,$$
whereas
$$
\dot P=\bigl[k_0x_0+k_1x_1\bigr](1-P)-\varepsilon_pP.
$$
On the boundary $x_0=0$, the first equation of system~\eqref{eq:fullmodel}
gives $\dot x_0=0$. If $x_1=0$, then
\[
\dot x_1=\mu\gamma x_0R(t)\bigl(1-X(t)\bigr)\ge0.
\]
Similarly, if $p_i=0$, $i=0,1$, the corresponding equations yield
\[
\dot p_i=k_ix_i\bigl(1-P(t)\bigr)\ge0.
\]
Finally, on the boundary $X(t)=1$, one has $\dot X(t)=-\varepsilon<0,$ whereas on the boundary $P(t)=1$,
$\dot P(t)=-\varepsilon_p<0.$
Therefore, the vector field points inward, or is tangent, on every boundary component of $\mathcal D$. By the standard tangency criterion for retarded functional differential equations, it follows that $\mathcal D$ is positively invariant.
Let $b=1+\gamma(\mu-2),$ $Q(t)=R(t)(1-X(t))$ and define
$$
\Lambda(t)=\gamma\mu x_0(t)+\bigl[1+\gamma(\mu-2)\bigr]x_1(t).
$$
Differentiating along solutions yields
$$
\dot \Lambda
=
\gamma\mu\left[\gamma(1-\mu)x_0Q-\varepsilon x_0\right]
+b\left[\bigl(\gamma\mu x_0+(1-\gamma)x_1\bigr)Q-\varepsilon x_1\right].
$$
Since $\gamma(1-\mu)+b=1-\gamma,$ we obtain
$
\dot \Lambda=\left[(1-\gamma)R(t)(1-X(t))-\varepsilon\right]\Lambda.
$
Therefore,
$$
\Lambda(t)=\Lambda(0)\exp\left(\int_0^t\left[(1-\gamma)R(s)(1-X(s))-\varepsilon\right]ds\right).
$$
Because $\boldsymbol\phi(0)\in\mathcal A$, we have $\Lambda(0)=0$. Hence,
$\Lambda(t)=0$ for all $t\geq0.$
Together with the positive invariance of $\mathcal D$, this implies
$\mathbf x(t)\in\mathcal A$ for all $t\geq0.$

Finally, let $t\geq0$ and $\theta\in[-\tau,0]$. If $t+\theta\geq0$, then the preceding argument gives $\mathbf x(t+\theta)\in\mathcal A$. If $t+\theta<0$, then $\mathbf x(t+\theta)=\boldsymbol\phi(t+\theta)\in\mathcal A$, since the initial history takes values in $\mathcal A$. Therefore, $\mathbf x_t(\theta)=\mathbf x(t+\theta)\in\mathcal A$ for every $\theta\in[-\tau,0]$, and hence $\mathbf x_t\in C([-\tau,0],\mathcal A)$ for every $t\geq0$. Thus, $C([-\tau,0],\mathcal A)$ is positively invariant.
\end{proof}

\begin{Remark}
When $\tau_0=\tau_1=0$, system~\eqref{eq:fullmodel} reduces to an ordinary differential equation, and the phase space $C([-\tau,0],\mathcal D)$ is naturally identified with $\mathcal D$. In this case, uniqueness allows solutions to be continued backward on their maximal intervals of existence. Since the identity defining $\mathcal A$ is preserved in both time directions wherever the solution exists, $\mathcal A$ is invariant under the flow of the non-delayed system. Thus, Proposition~\ref{prop:1} recovers Proposition~\ref{prop:1R} when $\tau_0=\tau_1=0$.
\end{Remark}

The previous results show that the introduction of delays does not modify the equilibrium equations and preserves the invariant set $\mathcal A$. Their effect is therefore dynamical rather than equilibrium-algebraic: although the equilibrium coordinates remain unchanged, the delays may alter transient behavior, local stability, oscillatory responses, and the global organization of the basins of attraction. We now illustrate these delay-induced effects.

Figure~\ref{fig:2basins}, computed numerically, reveals that the delay parameter $\tau$ does not merely slow down convergence, but rather changes the global organization of trajectories in phase space. As $\tau$ increases, the basin of attraction of the coexistence equilibrium shrinks, whereas the basin leading to extinction expands. Therefore, for the same values of $\mu$, $\gamma$, and the remaining biological parameters, initial conditions that converge to persistence when $\tau=0$ may instead evolve toward extinction for sufficiently large delays. This shows that replication delay constitutes an alternative extinction mechanism, not through changes in mutation pressure or replicative fitness, but through a delay-induced reorganization of the basins of attraction.
\begin{figure}
\captionsetup{width=\linewidth}
\begin{tikzpicture}[scale=0.8]

% Imagen 1
\node[inner sep=0pt] (caso1) at (0,0)
{\includegraphics[scale=0.35]{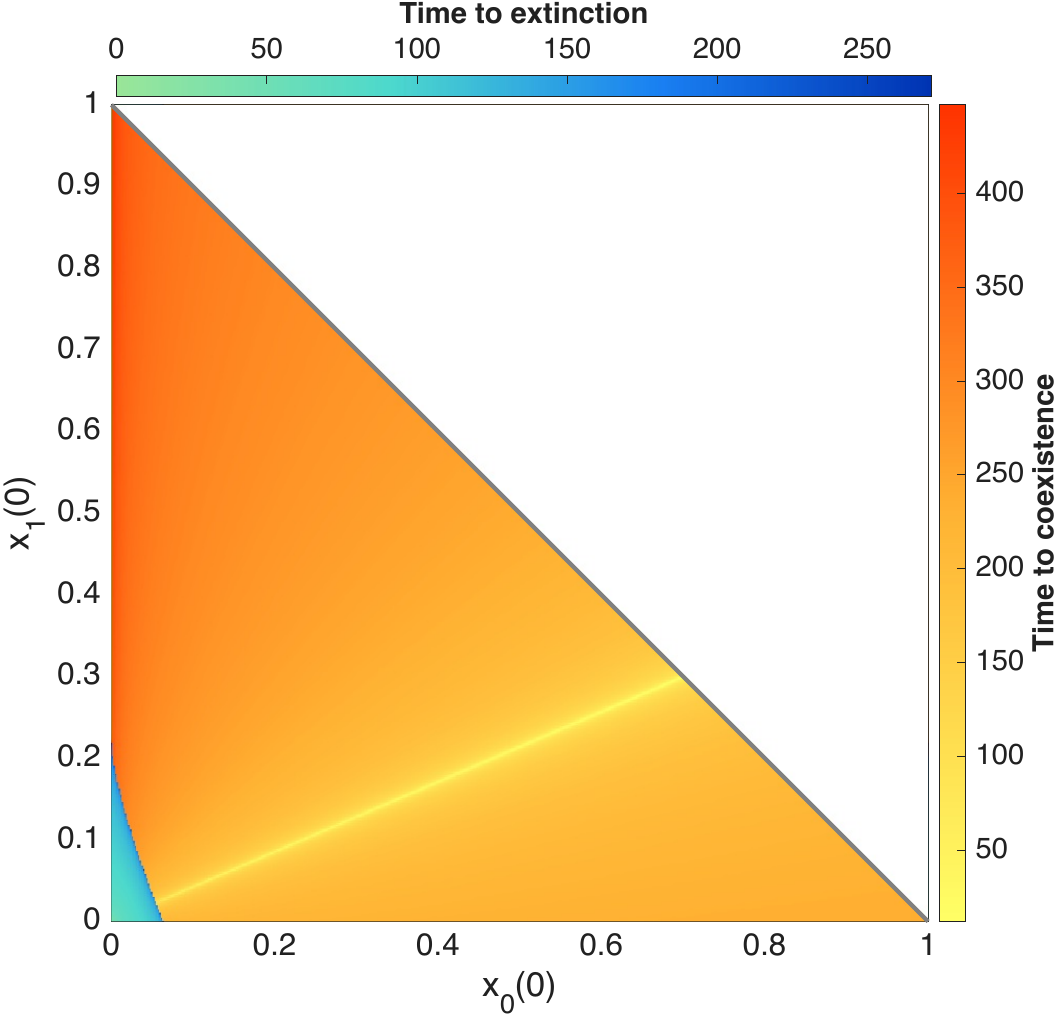}};

% Imagen 2
\node[inner sep=0pt] (caso2) at (9.5,0)
{\includegraphics[scale=0.35]{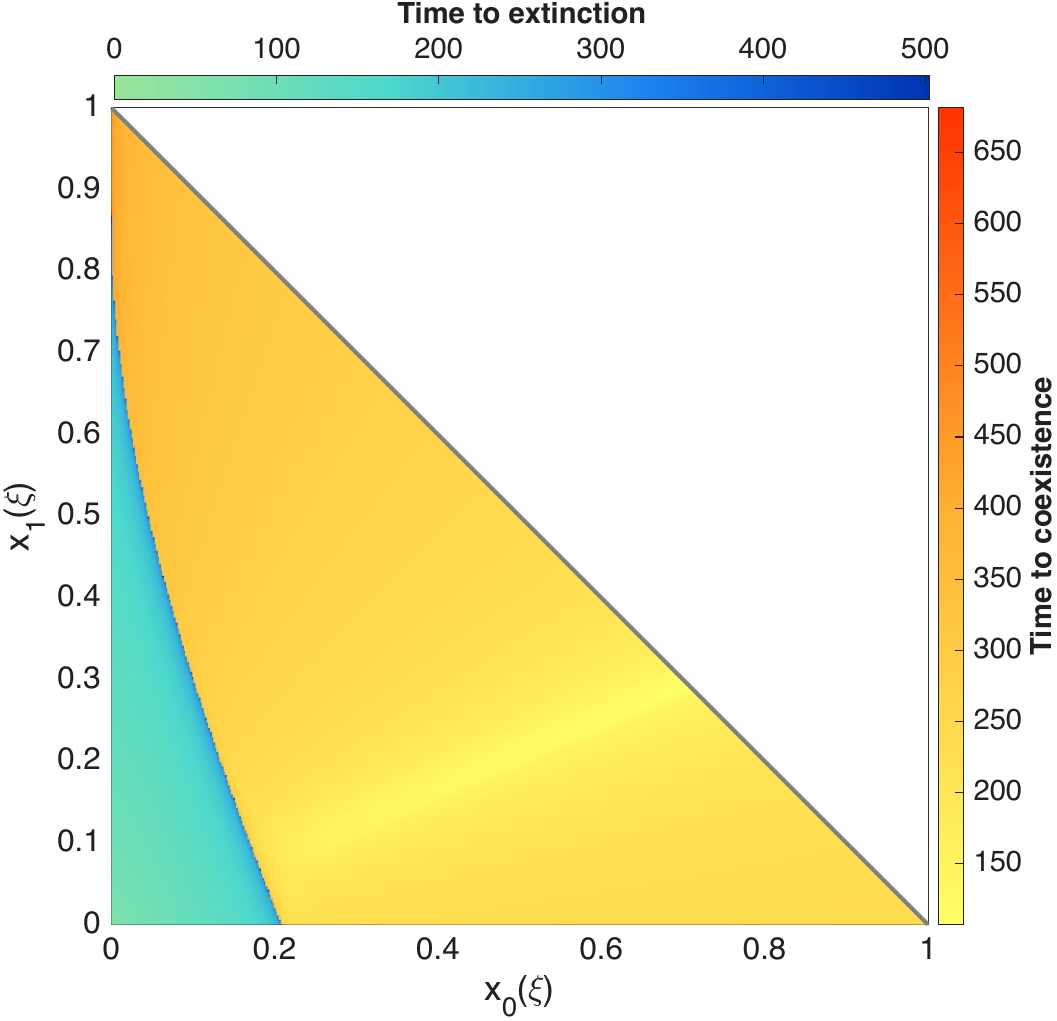}};

% Imagen 3
\node[inner sep=0pt] (caso3) at (0,-8)
{\includegraphics[scale=0.35]{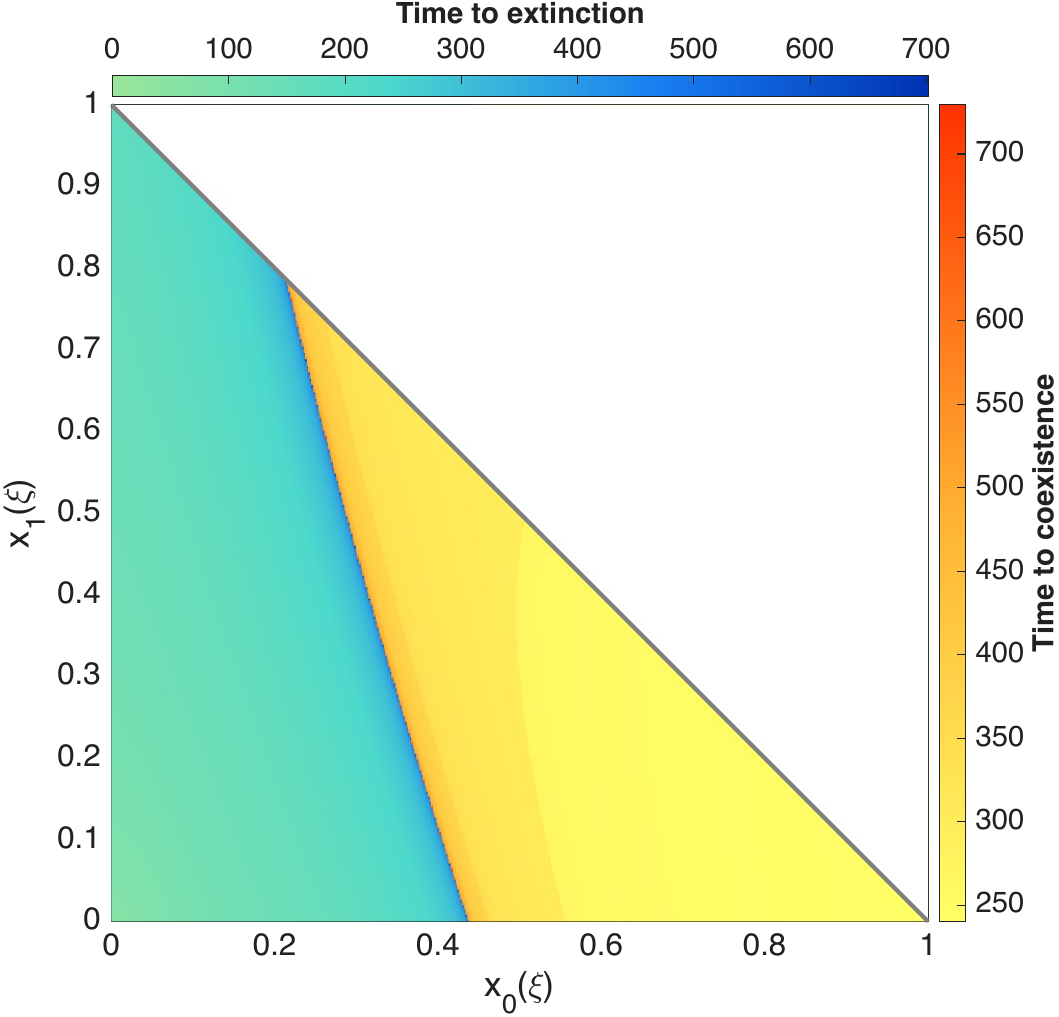}};

% Imagen 4
\node[inner sep=0pt] (caso3) at (9.8,-5.35)
{\includegraphics[scale=0.35]{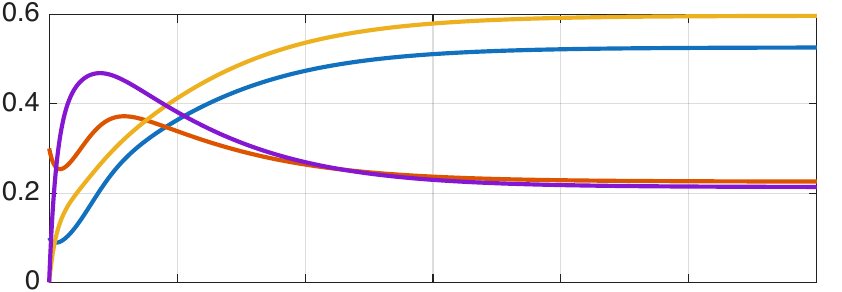}};
\node[inner sep=0pt] (caso3) at (9.8,-7.5)
{\includegraphics[scale=0.35]{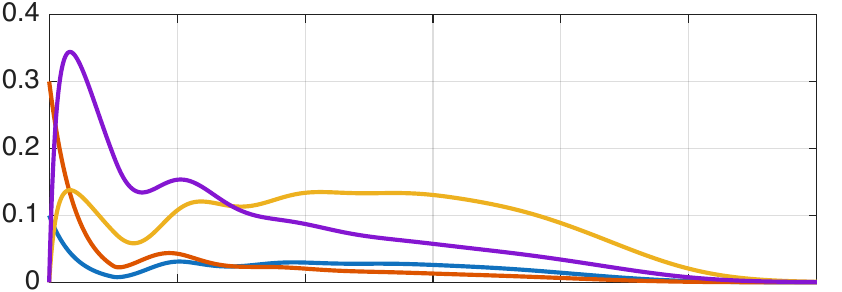}};
\node[inner sep=0pt] (caso3) at (9.8,-9.75)
{\includegraphics[scale=0.35]{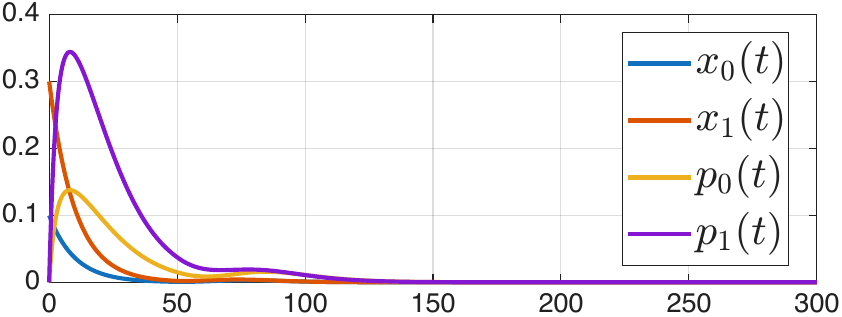}};

% Etiquetas (primer panel)
\coordinate [label=below:{{\textcolor{black}{{\bf a.}}}}] (E2) at (-4,3.9);
\coordinate [label=below:{{\textcolor{black}{\footnotesize $\tau=0$}}}] (E2) at (0,2.8);
\node at (-2.6,-1.3) {$\ast$};

% Etiquetas (segundo panel)
\coordinate [label=below:{{\textcolor{black}{{\bf b.}}}}] (E2) at (5,3.9);
\coordinate [label=below:{{\textcolor{black}{\small $\tau=25$}}}] (E2) at (9.5,2.8);
\node at (6.9,-1.3) {$\ast$};

% Etiquetas (tercer panel)
\coordinate [label=below:{{\textcolor{black}{{\bf c.}}}}] (E2) at (-4,-4.1);
\coordinate [label=below:{{\textcolor{black}{\footnotesize $\tau=50$}}}] (E2) at (0,-5.3);
\node at (-2.6,-9.25) {$\ast$};

% Etiquetas (cuarto panel)
\coordinate [label=below:{{\textcolor{black}{{\bf d.}}}}] (E2) at (5,-4.1);
\coordinate [label=below:{{\textcolor{black}{\scriptsize$\tau=0$}}}] (E2) at (9.8,-4.8);
\coordinate [label=below:{{\textcolor{black}{\scriptsize $\tau=25$}}}] (E2) at (9.8,-6.9);
\coordinate [label=below:{{\textcolor{black}{\scriptsize$\tau=50$}}}] (E2) at (9.8,-9.1);
\coordinate [label=below:{\rotatebox{90}{\textcolor{black}{\footnotesize state of variables}}}] (E2) at (6.3,-6.0);
\coordinate [label=below:{{\textcolor{black}{\footnotesize time}}}] (E2) at (9.8,-10.9);

\end{tikzpicture}
\caption{(a--c) Basins of attraction of system~\eqref{eq:fullmodel} for $p_0(\xi)=p_1(\xi)=0$ for every $\xi\in[-\tau,0]$ for $\tau=0$, $25$, and $50$, respectively. Blue indicates convergence to extinction and orange indicates convergence to coexistence. Increasing $\tau$ enlarges the extinction basin and reduces the basin of coexistence, showing that delay may act as an alternative route to extinction without changing mutation or fitness parameters. (d) Time series associated with the asterisked initial conditions. The parameter values are set as in Fig.~\ref{fig:Gradiente4D}.}
\label{fig:2basins}
\end{figure}
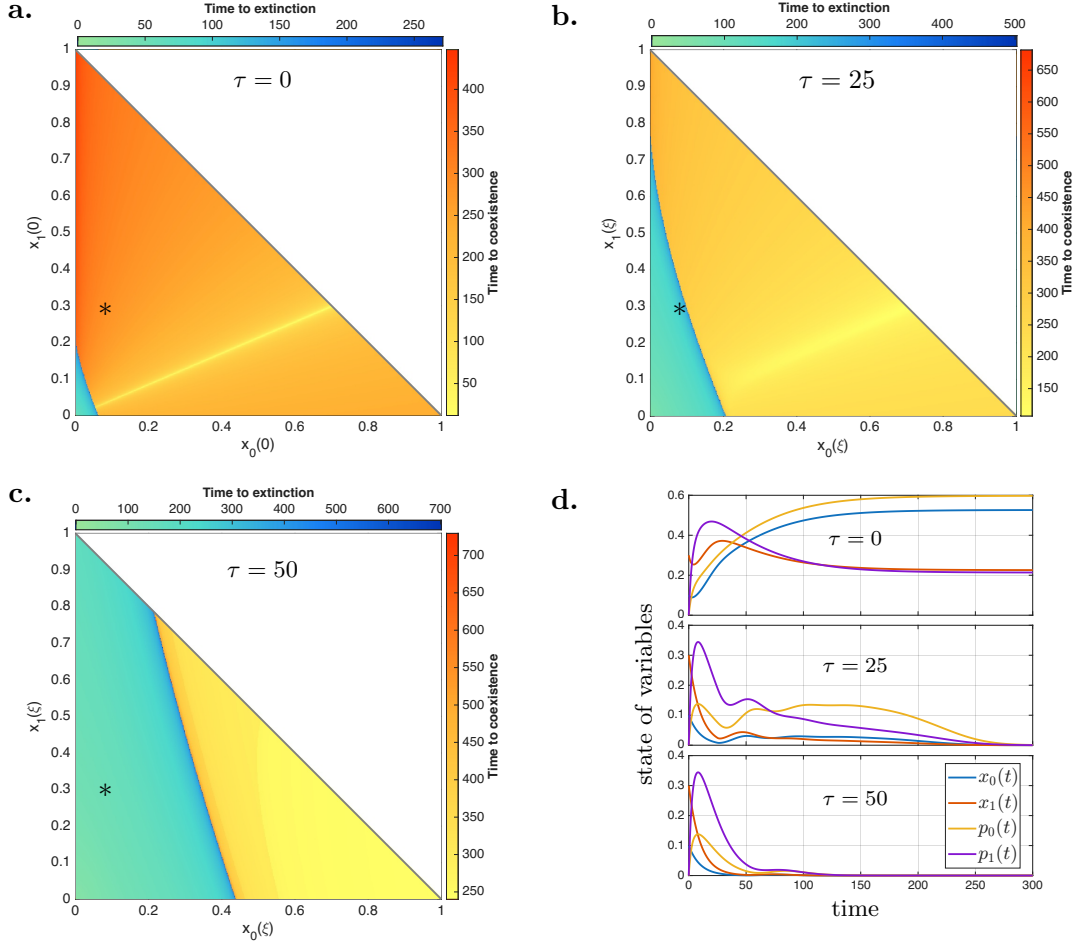

This phenomenon was previously identified in the three-variable delayed quasispecies model with functional complementation studied by Turner et al.~\cite{TurnerPRE}, which we extend and analyze in the next section. In that model, increasing the delay redirects trajectories from viral persistence toward extinction through a reorganization of the basins of attraction, a mechanism termed $\tau$-tipping. Remarkably, the same qualitative phenomenon is recovered here in the full four-variable model, despite its different equilibrium structure and the ability of both genome classes to produce functional RdRps. This result suggests that $\tau$-tipping is a structurally robust mechanism in polymerase-mediated quasispecies dynamics, with $\tau$ acting as an independent dynamical control parameter capable of triggering extinction through basin reorganization and attractor selection.

From a virological perspective, $\tau$-tipping describes failure to establish a self-sustaining intracellular replication system. During the interval before polymerase contributes effectively to RNA amplification, genomes continue to be lost through degradation. If the population crosses the basin boundary before sufficient replication capacity accumulates, subsequent polymerase production cannot rescue it. The effect is therefore not equivalent to slower exponential growth: it is a history-dependent transition between alternative long-term outcomes.

\begin{figure}[ht]
\centering
\captionsetup{width=\linewidth}
\begin{minipage}{0.38\textwidth}
\centering
\includegraphics[width=\textwidth]{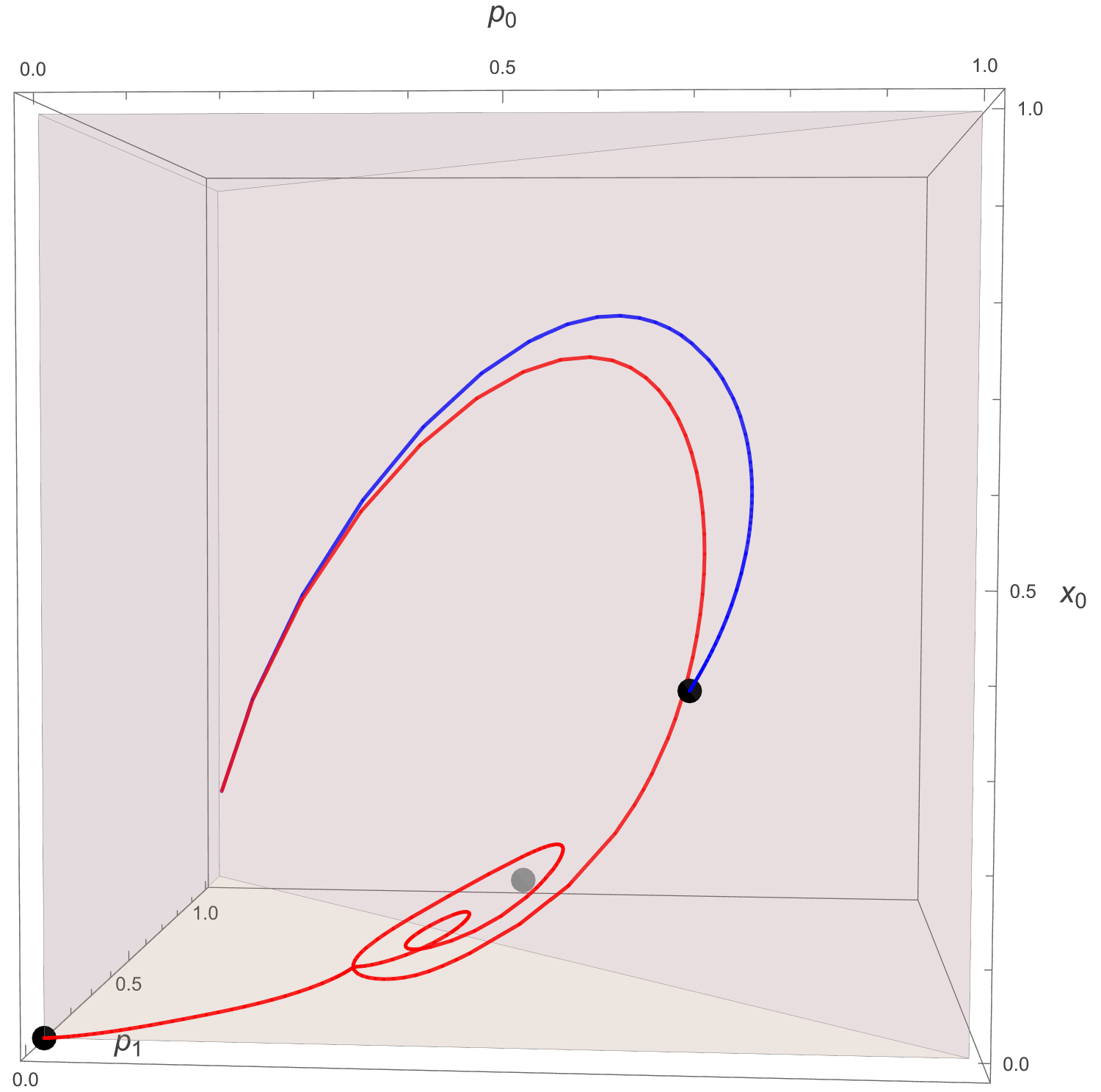}

\end{minipage}
\qquad\qquad
\begin{minipage}{0.38\textwidth}
\centering
\includegraphics[width=\textwidth]{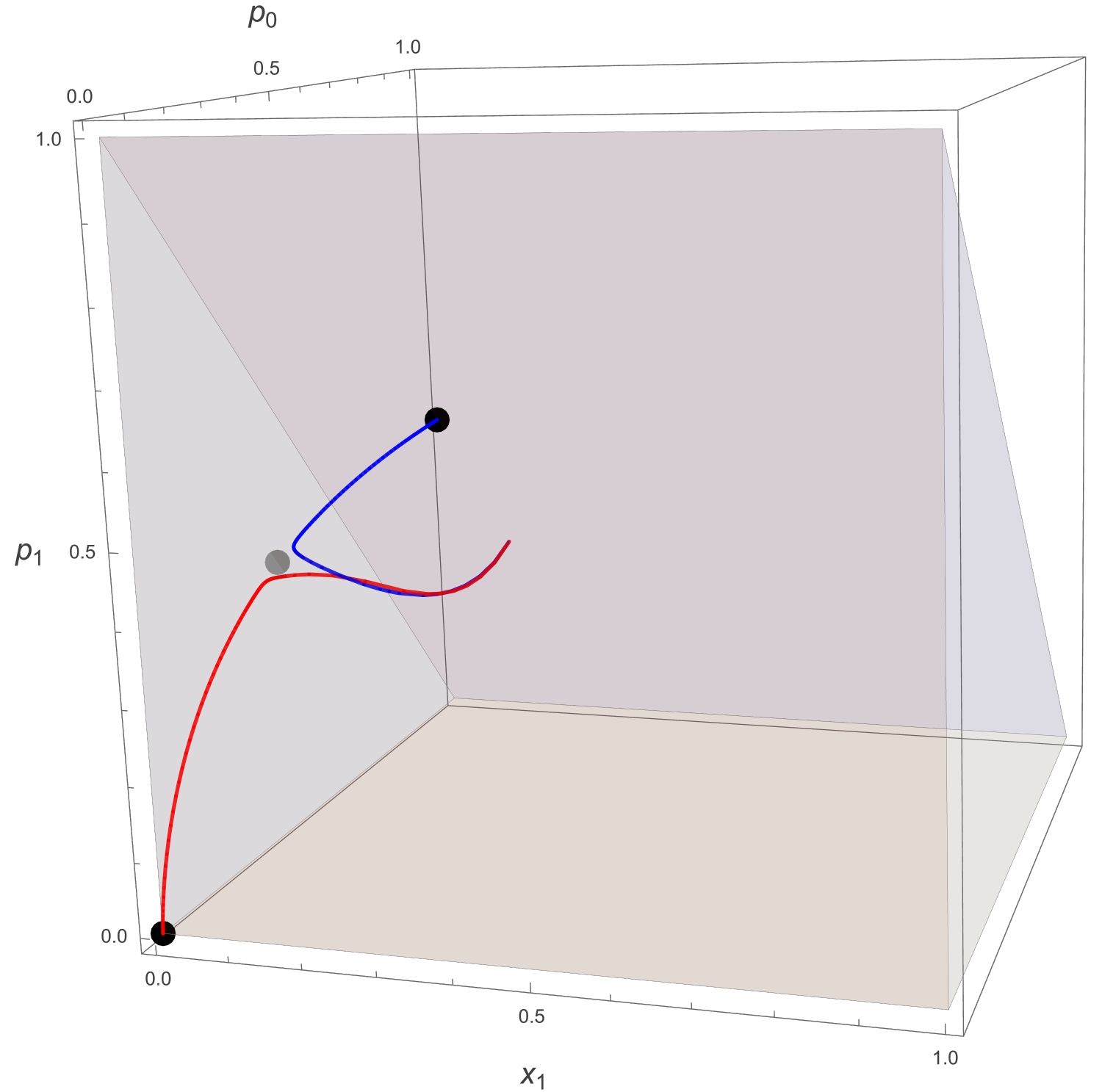}
\end{minipage}

\caption{Comparison of the delayed dynamics on two coordinate projections. The transparent polyhedron represents the biologically feasible region, while the black and gray points denote the equilibrium states of the reduced systems. (a) Projection of the trajectories onto the $(x_0,p_0,p_1)$ coordinates showing trajectories for $\tau=0$ (blue) and $\tau=50$ (red). The simulations were performed using $\mu=0.3$, $\gamma=0.334$, $r_0=1$, $r_1=0.7$, $\varepsilon=0.1$, $k_0=0.6$, and constant initial histories $x_0(\xi)=0.2$, $p_0(\xi)=0.1$, and $p_1(\xi)=0.5$ for $\xi\in[-\tau,0]$. (b) Projections onto the $(x_1,p_0,p_1)$ coordinates showing trajectories for $\tau=0$ (blue) and $\tau=10$ (red). The simulations were performed using $\mu=0.3$, $\gamma=0.964$, $r_0=1$, $r_1=7$, $\varepsilon=0.1$, $k_0=0.6$, $k_1=0.5$, and constant initial histories $x_1(\xi)=0.3$, $p_0(\xi)=0.5$, and $p_1(\xi)=0.4$ for $t\leq0$. In both projections, the inclusion of a sufficiently large delay induces a qualitative change in the long-term dynamics, shifting the system from a stable coexistence equilibrium to an extinction equilibrium.}
\label{fig:comparison_delays}
\end{figure}

%%%%%%%%%% aqui terminana FULL MODEL

\subsection{The reduced complementation model: equilibrium and bifurcation structure}
\label{sec:complementation}
Building upon the theoretical framework established in the previous sections, we now revisit the reduced complementation model introduced in Ref.~\cite{TurnerPRE} in order to obtain a systematic characterization of its equilibrium and bifurcation structure and to compare it directly with the full model analyzed above. In contrast with system (1), where both genome classes encode functional polymerases, the reduced system describes a defective mutant class that cannot produce its own RdRp and therefore depends on polymerase supplied in trans by the master genome. This formulation captures a biologically relevant complementation mechanism in which defective viral genomes lacking the RdRp genes depend on the replicative machinery supplied by viable master genomes.

As in the full model, we incorporate a time lag to represent the intracellular time required for RdRp synthesis, maturation, and activation. Thus, the model keeps the essential delayed polymerase-mediated replication mechanism, but removes the mutant-encoded polymerase variable. In this sense, the complementation model may be regarded as a reduction of the complete system~\eqref{eq:fullmodel}, obtained by focusing on the variables and interactions directly involved in the dynamics of defective genomes. Whereas Ref.~\cite{TurnerPRE} used this model to identify the \(\tau\)-tipping mechanism, the analysis below provides its complete equilibrium classification, analytical bifurcation conditions, and comparison with the autonomous-mutant model.

Let us consider system~\eqref{eq:DIviruses} and $p:= p_0.$ The biologically relevant domain is the submanifold of $\mathcal D$ given by
\begin{equation}\nonumber
\tilde{\mathcal D}
=
\left\lbrace
(x_0,x_1,p)\in\mathbb R^3:
x_0\geq 0,\ x_1\geq 0,\ 0\leq p\leq 1,\ x_0+x_1\leq 1
\right\rbrace .
\end{equation}
Since $x_0=0$ is invariant and on each boundary of $\tilde{\mathcal D}$ the vector field of \eqref{eq:DIviruses} is negative, it follows that $\tilde{\mathcal D}$ is invariant under the flow, see Fig.~\ref{fig:1}. This system has three equilibrium points given by $\mathbf x_1^*=(0,0,0)$ (full extinction),and the pair  $\mathbf x^*_{2}=(x_{02},x_{12},p_{2})$ and $\mathbf  x^*_{3}=(x_{03},x_{13},p_{3})$ involved in the coexistence between the master genome and its mutants, with
\begin{equation}
    \begin{split}\label{eq:equilCom}
        x_{0(2,3)}&= \frac{\mp2 \varepsilon ^2 (\gamma  (\mu -2)+1)}{\mathfrak{c}\pm k_0 (\gamma  (\mu -2)+1) (\gamma  (\mu -1) r_0+\varepsilon )}, \\
        x_{1(2,3)}&= \frac{\pm2 \gamma  \mu  \varepsilon ^2}{\mathfrak{c}\pm k_0 (\gamma  (\mu -2)+1) (\gamma  (\mu -1) r_0+\varepsilon )},\\
        p_{2,3}&= \frac{\mp2 k_0 \varepsilon  (\gamma  (\mu -2)+1)}{ \mathfrak{c}\pm k_0 (\gamma  (\mu -2)+1) (\gamma  (\mu -1) r_0-\varepsilon )},
    \end{split}
\end{equation}
with $\mathfrak{c}=\sqrt{ (k_0 (\gamma  (\mu -2)+1) (\gamma  (\mu -1) r_0+\varepsilon ))^2+4 \gamma  r_0 \varepsilon ^2 ( 2\gamma-1)(1-\mu) k_0 (\gamma  (\mu -2)+1)}.$ Note that the upper sign in \eqref{eq:equilCom} corresponds to the equilibrium $\mathbf x_2^*$ while the lower
sign corresponds to $\mathbf x_3^*.$
\begin{comment}
\begin{table}[H]
\centering
\label{tab:equilibria_summary}
\begin{tabular}{cccc}
\toprule
\textbf{Equilibria} & \textbf{State Variables at Equilibrium} & \textbf{Interpretation}\\
\midrule
$\mathbf{x^*_1}$ & $x_0 = 0,\ x_1 = 0,\ p = 0$ & Full extinction\\
$\mathbf{x^*_2}$ & $x_0 > 0,\ x_1 > 0,\ p > 0$ & Coexistence\\
$\mathbf{x^*_3}$ & $x_0 > 0,\ x_1 > 0,\ p > 0$ & Coexistence\\
\bottomrule
\end{tabular}
\captionsetup{width=\linewidth}
\caption{Summary of equilibrium points and their biological interpretation in the system with complementation.}
\end{table}
\end{comment}
\begin{figure}
\centerline{\includegraphics[width=\textwidth]{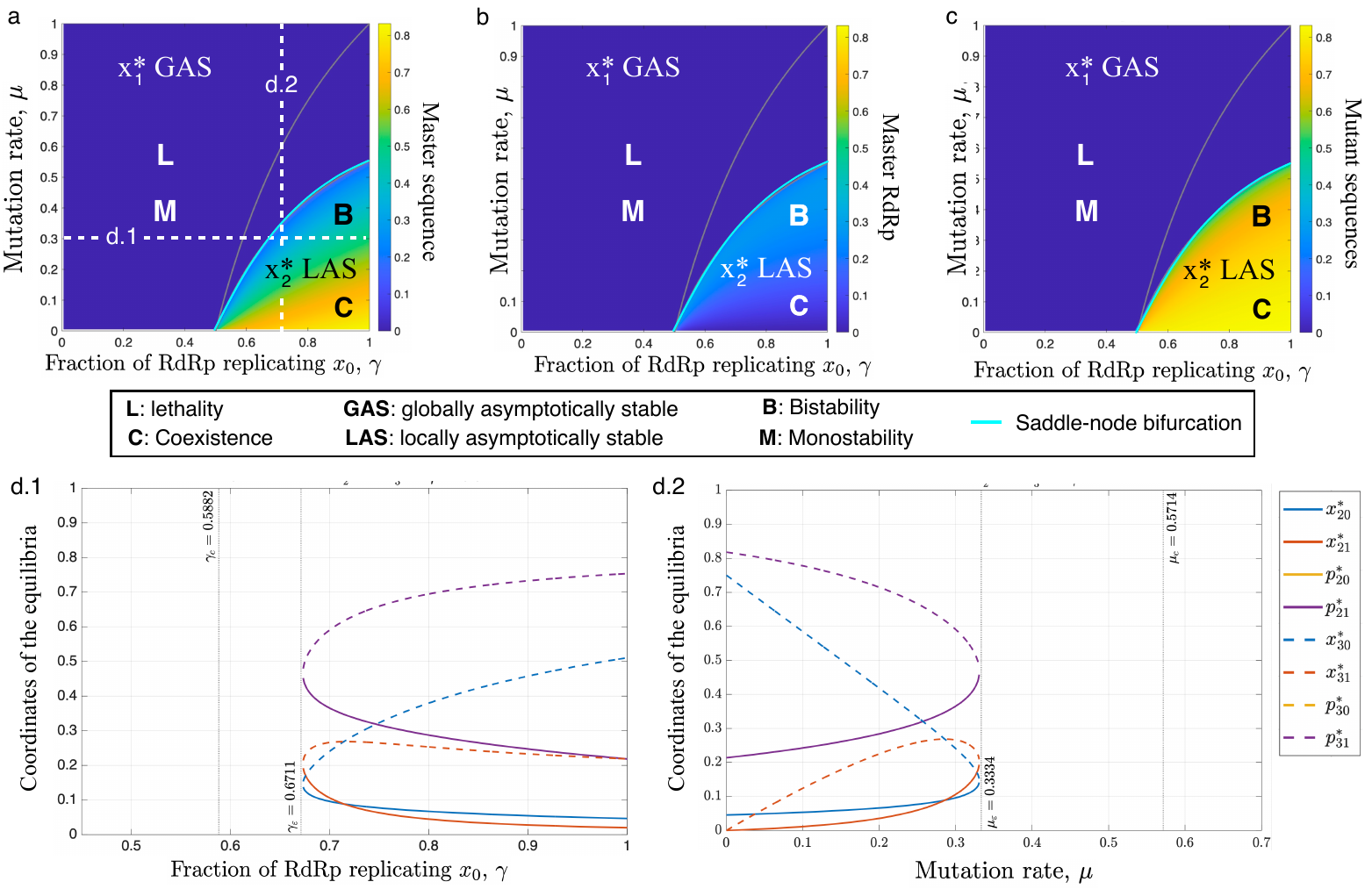}}
  \captionsetup{width=\linewidth}
\caption{Phase diagram of the master genome equilibrium of eqs.~\eqref{eq:DIviruses} for $\tau=0$, showing persistence at low $\mu$ and large $\gamma$. The blue area corresponds to the full collapse of the system. Parameters are fixed at $r_0=0.7,$ $k = 0.6$, and $\varepsilon=\varepsilon_p = 0.1$. The gray curves delimit the parameter region for which the persistence equilibria $\mathbf x_2^*$ and $\mathbf x_3^*$ are real and lie inside the biologically relevant set $\tilde{\mathcal D}$.}
\label{fig:3D_bif}		
\end{figure}
\vspace{-0.45cm}
\begin{figure}[h!]
\captionsetup{width=\linewidth}
\begin{center}
		\subfigure[$\gamma\leq\frac{1}{2-\mu}$]{\includegraphics[scale=0.23]{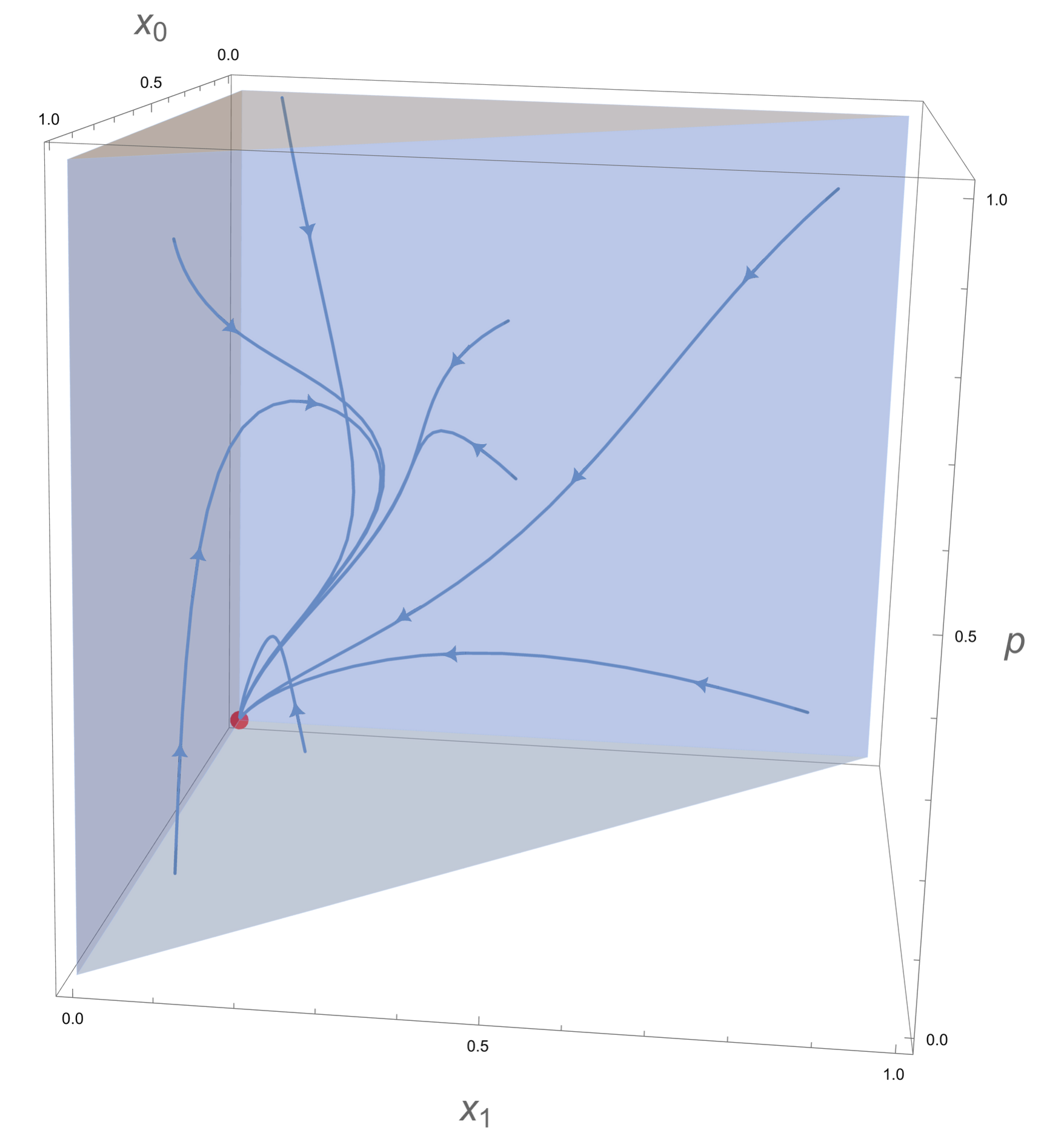}}\qquad
		\subfigure[$\gamma>\frac{1}{2-\mu}$]{\includegraphics[scale=0.3]{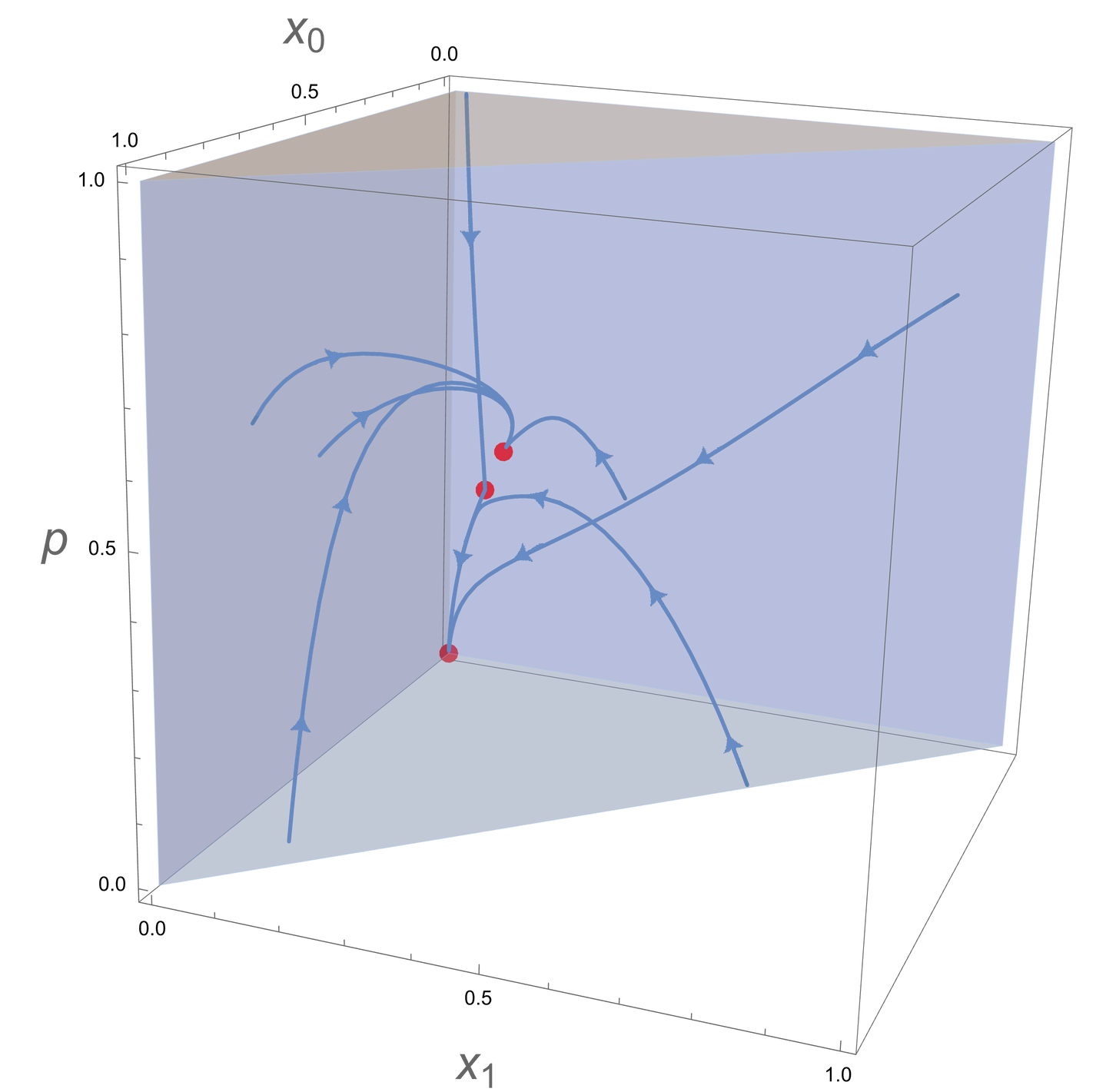}}
\end{center}
\vspace{-0.25cm}
\caption{Dynamics of the system \eqref{eq:DIviruses} in the biologically relevant region $\tilde{\mathcal D}$. (a) Phase portrait for $\gamma \leq \frac{1}{2-\mu}$ where the origin is the only asymtptotic state with $\mu=0.4$, $\gamma=0.6$, $k=0.5$, $\varepsilon=0.14$ and $\tau=0$. (b) Dynamics for $\gamma >\frac{1}{2-\mu}$, where the system has three equilibria and both extinction and persistence of the genomes is possible. The parameters are $\mu=0.4$, $\gamma=0.9$, $k=0.5$, and $\varepsilon=0.14$.}
\label{fig:1}		
\end{figure}
\begin{Proposition}\label{prop:x1yx22}
Let $\mathbf{x}_2^*$ and $\mathbf{x}_3^*$ be equilibria of the system. 
If the conditions $\mu < 2 - \frac{1}{\gamma}$ and
$$\varepsilon \leq \frac{-\gamma(1-\mu)k_0(1 - 2\gamma + \gamma\mu)r_0}{- k_0(1 - 2\gamma + \gamma\mu) + 2\sqrt{(1 - 2\gamma)\gamma r_0(1 - \mu)\, k_0(1 - 2\gamma + \gamma\mu)}}$$
are satisfied, then the equilibria $\mathbf{x}_2^*$ and $\mathbf{x}_3^*$ lie in the bounded region $\tilde{\mathcal D}$.
\end{Proposition}

Biologically, Proposition~\ref{prop:x1yx22} shows that complementation can maintain defective genomes only as long as the master genome remains sufficiently viable. Once the mutation-affinity balance crosses the critical threshold, the source of functional RdRp collapses, and defective genomes cannot persist autonomously. Therefore, the threshold does not merely mark the loss of the master sequence; in the DVG model it marks the onset of total population extinction. In classical quasispecies models, disappearance of the master genome can leave a self-replicating mutant cloud. In the present DVG limit, by contrast, replication competence is a collective property of the complete-genome/DVG system. Because only one class supplies polymerase, the mutant population cannot outlive the loss of its helper. Thus, functional dependence converts a compositional transition into an ecological collapse of the intracellular replicating community.

Proposition~\ref{prop:x1yx22} also shows that the persistence of the master genome under complementation is constrained by a threshold condition involving the mutation and affinity parameters. When genome degradation is incorporated through the parameter $\varepsilon$, this admissibility criterion becomes more restrictive: only a smaller set of parameter values leads to biologically meaningful persistence equilibria. The corresponding feasible domain is illustrated in Fig.~\ref{fig:3D_bif}.

%Figure~\ref{fig:3D_bif} shows that, in the complementation model, increasing $\mu$ or decreasing $\gamma$ can drive the system from viral persistence to complete lethality. Importantly, this transition does not correspond to the classical error threshold: because mutant genomes cannot produce functional RdRp, loss of the master genome ultimately implies extinction of the entire population. Accordingly, the catastrophe to lethality is mediated by a saddle-node bifurcation.

%Figures~\ref{fig:1} and~\ref{fig:5} then explore the dynamical consequences of introducing time lags into the system. These simulations show that, although the algebraic location of the equilibria is not changed by the delay, the transient approach to the final state can be strongly modified by $\tau$. 

Figure~\ref{fig:3D_bif} shows that, in the complementation model, increasing $\mu$ or decreasing $\gamma$ can drive the system from viral persistence to complete lethality. Importantly, this transition does not correspond to the classical error threshold: because mutant genomes cannot produce functional RdRp, loss of the master genome ultimately implies extinction of the entire population. Accordingly, the catastrophe to lethality is mediated by a saddle-node bifurcation. The corresponding phase portraits further illustrate the geometry of this transition (Figs.~\ref{fig:1}-\ref{fig:6}). For $\gamma\leq(2-\mu)^{-1}$, the extinction equilibrium is the global attractor and all trajectories eventually converge to the origin (Figs.~\ref{fig:1}a), whereas for $\gamma>(2-\mu)^{-1}$ additional equilibria emerge through the saddle-node bifurcation, allowing for viral coexistence depending on the initial condition (Fig.~\ref{fig:1}b). Figure~7 shows that, within the extinction regime, trajectories converge to the origin for different initial conditions and values of $\tau$, although the delay modifies their transient dynamics. In contrast, close to the persistence threshold, Fig.~\ref{fig:6} reveals a strong dependence on both the initial condition and the delay: increasing $\tau$ can redirect trajectories from the coexistence attractor toward complete extinction. Thus, as in the full model, delayed RdRp availability can induce $\tau$-tipping by reorganizing attractor selection without modifying the equilibrium structure.

\begin{figure}[h!]
\captionsetup{width=\linewidth}
\begin{center}
		\subfigure[$x_0(\xi)=0.2,$ $x_1(\xi)=0.8$ and \phantom{(aa)} $p(\xi)=0.9,$ with $\xi\in\interval{-\tau}{0}.$]{\includegraphics[scale=0.4]{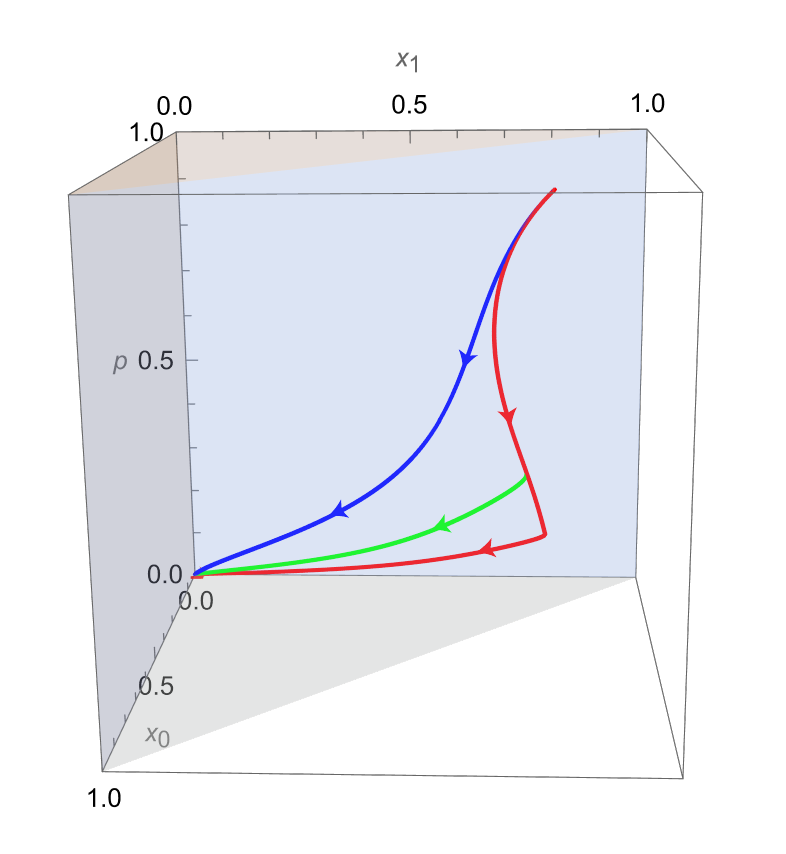}}
		\subfigure[$x_0(\xi)=1,$ $x_1(\xi)=0$ and \phantom{(ba)} $p(\xi)=0.5,$ with $\xi\in\interval{-\tau}{0}.$]{\includegraphics[scale=0.4]{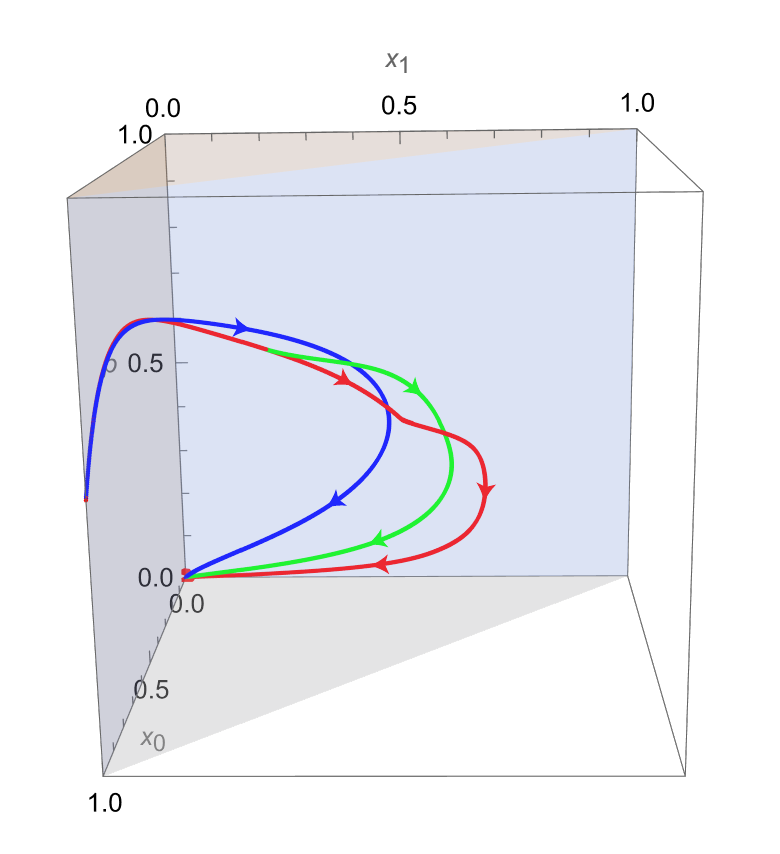}}
        \subfigure[$x_0(\xi)=0.6,$ $x_1(\xi)=0.3$ and \phantom{(ba)} $p(\xi)=0.1,$ with $\xi\in\interval{-\tau}{0}.$]{\includegraphics[scale=0.4]{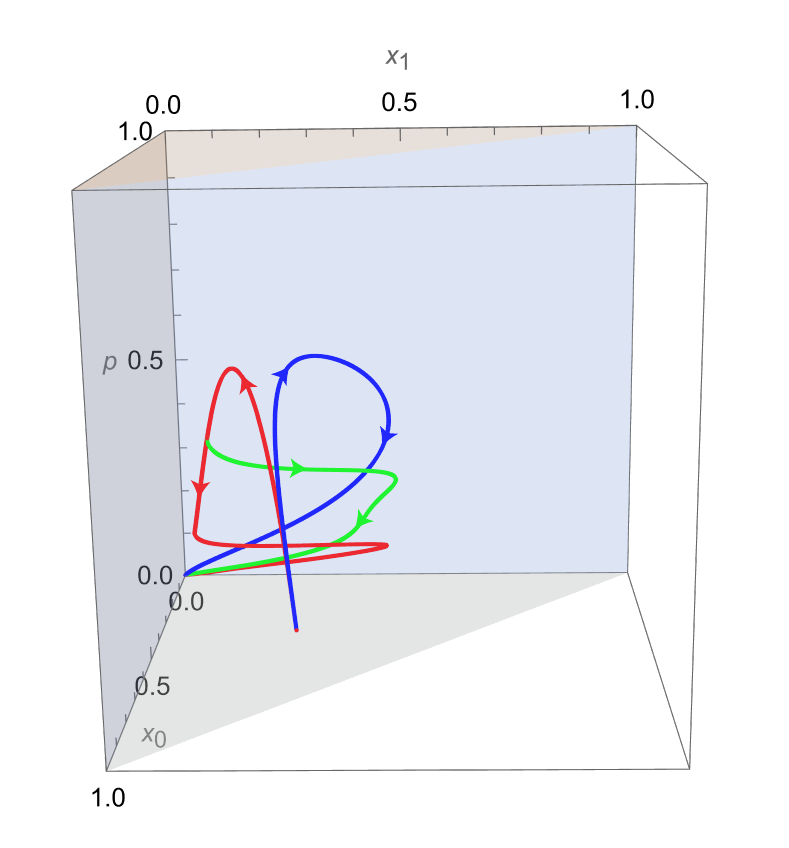}}
\end{center}
\vspace{-0.25cm}
\caption{System trajectories for three distinct initial conditions and three values of the delay parameter ($\tau = 0$, $25$, $50$; blue, green, and red curves, respectively) in the extinction regime, $\gamma\leq(2-\mu)^{-1}$, assuming sufficiently small genome degradation. The parameters considered are $\mu=0.2,$ $\gamma=0.4,$ $k=0.4$ and $\varepsilon=0.1$.}
\label{fig:5}		
\end{figure}

\begin{figure}[h!]
\captionsetup{width=\linewidth}
\begin{center}
		\subfigure[$x_0(\xi)=0.8,$ $x_1(\xi)=0.2$ and \phantom{(aa)} $p(\xi)=0,$ with $\xi\in\interval{-\tau}{0}.$]{\includegraphics[scale=0.395]{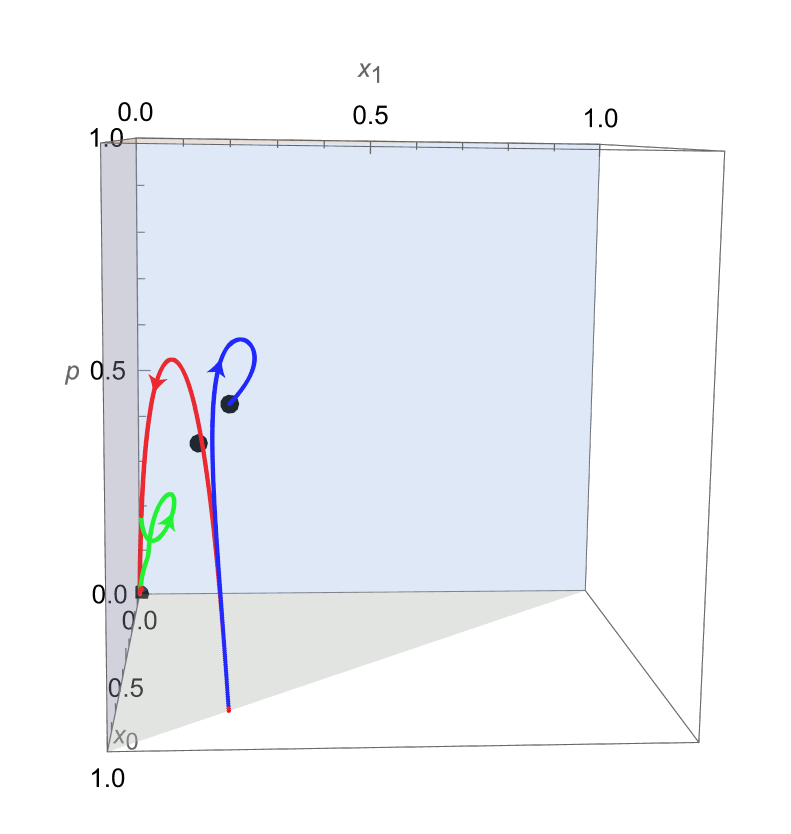}}
		\subfigure[$x_0(\xi)=0.1,$ $x_1(\xi)=0.8$ and \phantom{(aa)} $p(\xi)=0.9,$ with $\xi\in\interval{-\tau}{0}.$]{\includegraphics[scale=0.365]{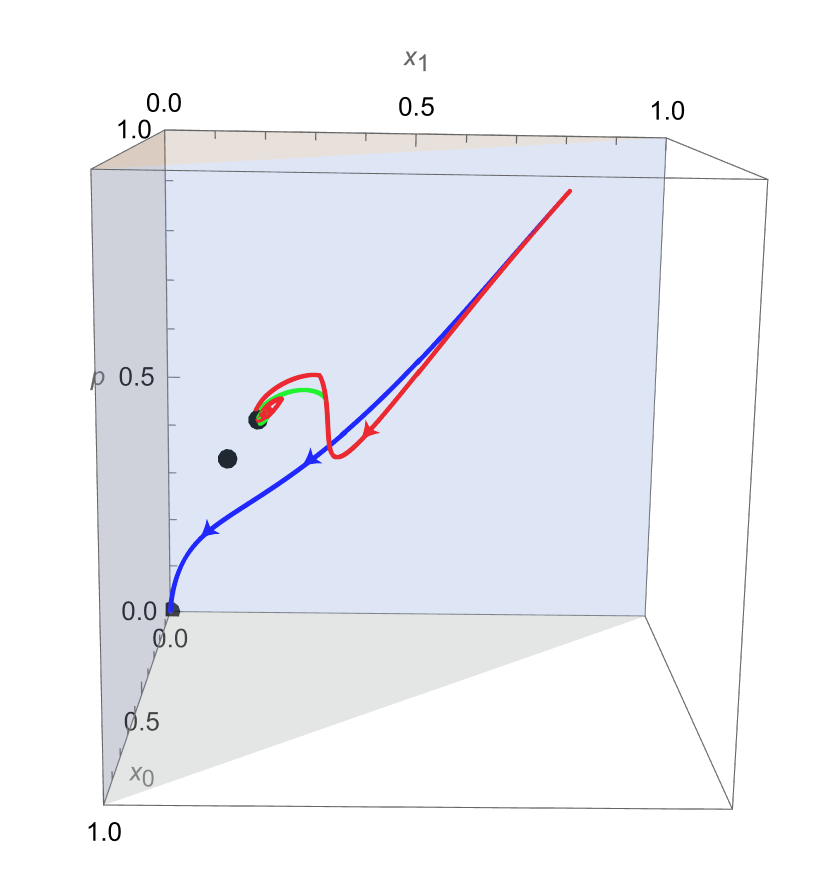}}
        \subfigure[$x_0(\xi)=0.4,$ $x_1(\xi)=0.6$ and \phantom{(aa)} $p(\xi)=0,$ with $\xi\in\interval{-\tau}{0}.$]{\includegraphics[scale=0.315]{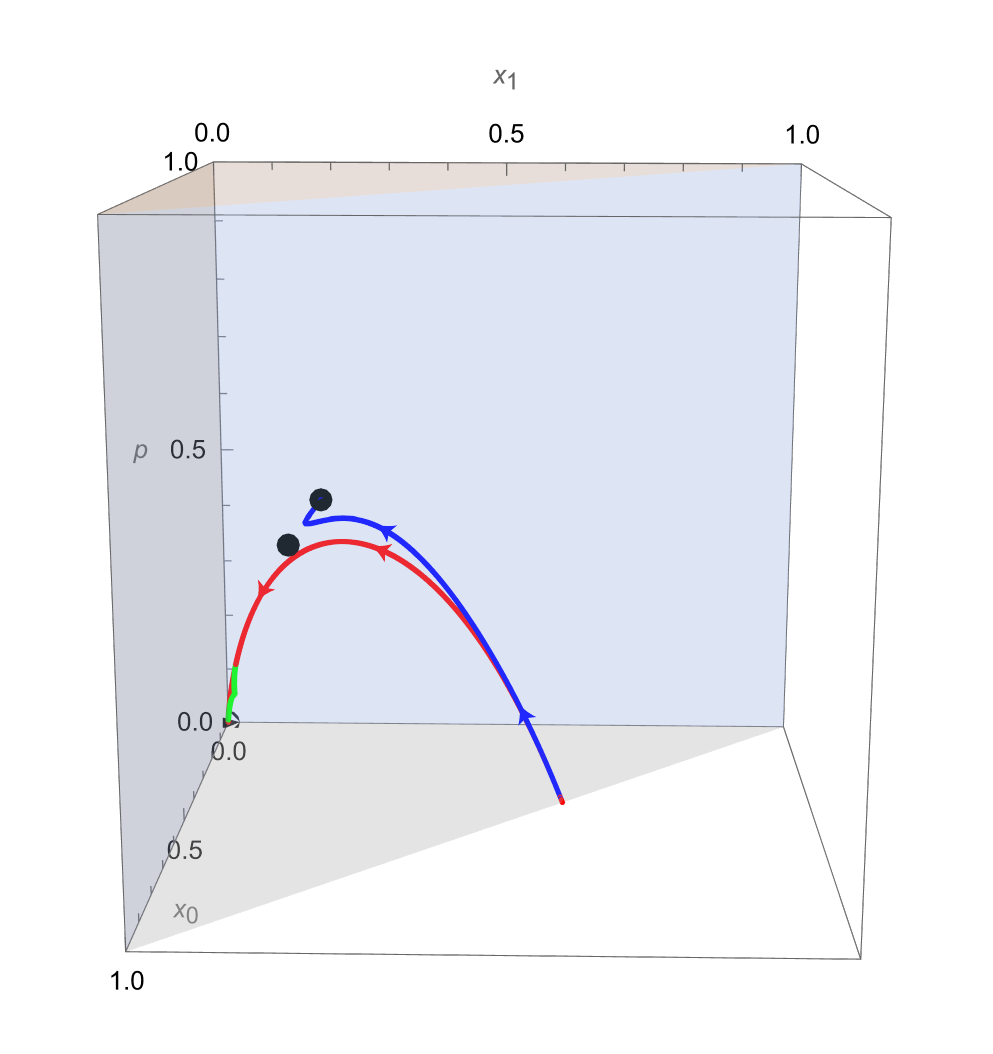}}
\end{center}
\vspace{-0.25cm}
\caption{System trajectories for three distinct initial conditions and three values of the delay parameter ($\tau = 0$, $25$, $50$; blue, green, and red curves, respectively), with parameters set just before the error threshold, $\gamma>(2-\mu)^{-1}$. In this regime, introducing a finite delay can lead trajectories to settle at entirely different equilibria, depending on both the initial condition and the value of $\tau$.}
\label{fig:6}		
\end{figure}

\subsubsection{Periodic orbits without delays}

In what follows, we set $\tau = 0$. Building on Proposition~\ref{prop:no-periodic-orbits}, it is straightforward to verify that system~\eqref{eq:DIviruses} also does not admit periodic orbits. This result is formalized in the following corollary.
\begin{Corollary}
For any $T > 0$, the system~\eqref{eq:DIviruses} does not admit any $T$-periodic solution with initial conditions in the region $\tilde{\mathcal{D}}$.
\end{Corollary}

\section{Conclusions}

We have developed and analyzed a polymerase-mediated quasispecies framework in which master and mutant genomes can encode distinct functional polymerases. This formulation extends the minimal complementation model previously used to identify the \(\tau\)-tipping mechanism~\cite{TurnerPRE} and produces a richer equilibrium structure involving master-mutant coexistence, mutant-only persistence, and complete extinction. We have shown how mutation, genome-polymerase template preference, replication efficiency, and degradation jointly organize these regimes through transcritical and saddle-node bifurcations.

A central result is that the \(\tau\)-tipping mechanism identified in Ref. ~\cite{TurnerPRE} persists in this higher-dimensional and functionally more general setting. Delays do not change the equilibrium coordinates, but they reorganize the basins of attraction. Consequently, populations that persist under instantaneous polymerase activity may become extinct when functional replication machinery becomes available after a sufficiently long interval. Recovering this transition when both genome classes produce functional polymerases shows that \(\tau\)-tipping is not restricted to complete polymerase dependence or to the particular reduced architecture considered previously.

The comparison between autonomous mutants and polymerase-defective genomes reveals an additional qualitative effect. When mutant genomes encode functional polymerases, loss of the master sequence can leave a viable mutant population. By contrast, when mutants depend on polymerase supplied in trans by the master genome, loss of the master also removes the source of replication machinery and converts master-sequence loss into extinction of the entire population governed by a saddle-node bifurcation. Thus, the functional organization of the mutant spectrum determines whether a quasispecies transition produces evolutionary replacement or population collapse.

These results suggest several virologically testable predictions. Viral persistence should depend not only on the amount of polymerase eventually produced, but also on how rapidly functional replication complexes are established relative to genome degradation. Perturbations that delay polyprotein processing, cofactor recruitment, membrane association, or replicase assembly may therefore have effects disproportionate to their influence on the final abundance or catalytic activity of polymerase. In multistable regimes, outcomes should also depend on the initial ratio of replication-competent genomes, defective genomes, and functional replication machinery.

The model is deliberately minimal and should not be interpreted as a quantitative representation of a particular virus. The delay is an effective representation of several intracellular processes, the mutant spectrum is aggregated into a single class, and stochasticity, spatial compartmentalization, innate immunity, packaging, and cell-to-cell transmission are omitted. Future work should test whether $\tau$-tipping persists under distributed delays, explicit maturation compartments, stochastic founding populations, cis-preferential replication, experimentally measured DVG distributions, and virus-specific replication kinetics.

Despite its deliberately minimal formulation, the model generates several experimentally testable predictions. First, equal endpoint polymerase abundance need not produce equal viral yields: delaying polymerase activation may drive the population to extinction even if polymerase levels subsequently reach those observed during a persistent infection. Second, this effect should be strongly nonlinear, because near the basin boundary a modest increase in the time required for replicase activation can cause an abrupt transition from persistence to extinction. Third, delayed replicase formation should interact synergistically with viral RNA degradation, such that RNA-destabilizing treatments strongly enhance the effect of a given delay. Fourth, inocula enriched in defective viral genomes should be particularly delay-sensitive when these genomes depend on helper-derived polymerase. Early complementation may therefore rescue the population, whereas the same amount of polymerase supplied later may fail to do so. More generally, infection outcome should be history-dependent: changing the temporal order of genome delivery and polymerase availability may alter persistence even when the final molecular composition is identical. These predictions translate the mathematical mechanism of basin reorganization into experiments involving pulse delivery, inducible polymerase-expression systems, protease or polymerase-cofactor inhibitors, RNA-destabilizing treatments, and controlled mixtures of defective and helper genomes. Such experiments could establish whether the timing of functional replicase availability constitutes an independently controllable determinant of viral persistence.

\appendix
\section{Technical proofs}

\subsection{Proof of Proposition 2}

Observe that the second component of the equilibrium $\mathbf{x}^*_2$, denoted by $x_{12}$, can be written as 
$$\frac{k_1((1-\gamma)r_1-\varepsilon)+\sqrt{k_1\left(k_1((1-\gamma)r_1-\varepsilon)^2 -4(1-\gamma)r_1\varepsilon^2\right)}}{2r_1k_1(1-\gamma)},$$
while the component $p_{12}$ is given by
$$\frac{k_1((1-\gamma)r_1+\varepsilon)
+\sqrt{k_1\left(k_1((1-\gamma)r_1-\varepsilon)^2
-4(1-\gamma)r_1\varepsilon^2\right)}}
{2r_1(k_1+\varepsilon)(1-\gamma)}.$$
Note that both expressions contain the same radical term. We first show that this quantity is real. Define
\begin{align*}
\mathcal N
&=k_1\left(k_1((1-\gamma)r_1-\varepsilon)^2-4(1-\gamma)r_1\varepsilon^2\right)\\
&=k_1\Big([k_1-4(1-\gamma)r_1]\varepsilon^2
-2k_1(1-\gamma)r_1\varepsilon
+k_1(1-\gamma)^2r_1^2\Big)\\
&=k_1\big(k_1-4(1-\gamma)r_1\big)
\left(\varepsilon-\frac{(1-\gamma)\sqrt{k_1}r_1}
{\sqrt{k_1}-2\sqrt{(1-\gamma)r_1}}\right)
\left(\varepsilon-\frac{(1-\gamma)\sqrt{k_1}r_1}
{\sqrt{k_1}+2\sqrt{(1-\gamma)r_1}}\right).
\end{align*}
We analyze the three possible cases:\\

\noindent (1) $k_1-4(1-\gamma)r_1>0,$ $\varepsilon-\frac{(1-\gamma)\sqrt{k_1}r_1}{\sqrt{k_1}-2\sqrt{(1-\gamma)r_1}}<0$ and $\varepsilon-\frac{(1-\gamma)\sqrt{k_1}r_1}{\sqrt{k_1}+2\sqrt{(1-\gamma)r_1}}\le0.$ Hence, $\mathcal N\geq 0.$\\

\noindent (2) $k_1 - 4(1-\gamma) r_1<0,$ $
\varepsilon - \frac{(1-\gamma)\sqrt{k_1}r_1}{ \sqrt{k_1}  - 2\sqrt{(1-\gamma)r_1}} >0$ and $\varepsilon-\frac{(1-\gamma)\sqrt{k_1}r_1}{ \sqrt{k_1}+2\sqrt{(1-\gamma)r_1}}\leq0.$ Hence, $\mathcal{N}\geq0.$\\

\noindent (3) $k_1 - 4(1-\gamma) r_1=0.$ Hence, $\mathcal N=-2k_1(1-\gamma)r_1\varepsilon +k_1(1-\gamma)^2r_1^2 =k_1(1-\gamma)r_1\big(r_1(1-\gamma)-2\varepsilon\big)\ge0.$\\

\noindent Therefore, in all cases the coordinates of $\mathbf{x}_2^*$ are real. The same argument applies to $\mathbf{x}_3^*$. Next we prove that $\mathbf{x}_2^*$ and $\mathbf{x}_3^*$ lie in the bounded region $\mathcal D$.  
In cases (1) and (2) we have
$$\varepsilon \le\frac{\sqrt{k_1}r_1(1-\gamma)}{\sqrt{k_1}+2\sqrt{r_1(1-\gamma)}}<r_1(1-\gamma),$$
and in case (3) the same inequality follows directly, $\varepsilon<r_1(1-\gamma).$ Using this bound we obtain
\begin{align*}
0&<
\frac{k_1((1-\gamma)r_1-\varepsilon)
-\sqrt{(k_1((1-\gamma)r_1-\varepsilon))^2
-4(1-\gamma)k_1r_1\varepsilon^2}}
{2r_1k_1(1-\gamma)}\\
&<
\frac{k_1((1-\gamma)r_1-\varepsilon)
+\sqrt{(k_1((1-\gamma)r_1-\varepsilon))^2}}
{2r_1k_1(1-\gamma)}
<1.
\end{align*}
Similarly,
\begin{align*}
0&<
\frac{k_1((1-\gamma)r_1+\varepsilon)
-\sqrt{(k_1((1-\gamma)r_1-\varepsilon))^2
-4(1-\gamma)r_1\varepsilon^2}}
{2r_1(k_1+\varepsilon)(1-\gamma)}\\
&<
\frac{k_1((1-\gamma)r_1+\varepsilon)
+\sqrt{(k_1((1-\gamma)r_1-\varepsilon))^2}}
{2r_1(k_1+\varepsilon)(1-\gamma)}\\
&=
\frac{k_1}{k_1+\varepsilon}
<1 .
\end{align*}
Hence all coordinates of $\mathbf{x}_2^*$ and $\mathbf{x}_3^*$ belong to the interval $(0,1)$. Consequently, both equilibria lie in the bounded region $\mathcal D$. 
\hspace*{\fill}$\square$\bigskip\noindent

\subsection{Proof of Proposition 3}

Since the coordinates of $\mathbf{x}_4^*$ and $\mathbf{x}_5^*$ involve a common radical term, we first verify that this quantity is real. Indeed, the expression under the square root can be written as
      {\small\begin{align*}
      \mathrm{N} =\left[ \left(P_1 + 2\sqrt{(1-2\gamma ) \gamma(1-\mu ) Q_0}\right) \varepsilon - \gamma(1-\mu ) Q_0 \right] \left[ \left(P_1 - 2\sqrt{(1-2\gamma ) \gamma(1-\mu ) Q_0}\right) \varepsilon - \gamma(1-\mu ) Q_0 \right],
  \end{align*}  }
where, $P_1 = k_0(1 - 2\gamma + \gamma\mu) - k_1 \gamma \mu$, $Q_0= k_0(1 - 2\gamma + \gamma\mu) - k_1 r_1 \gamma \mu$. Assuming that
$$ k_1 r_1\leq k_0\quad\quad   \mu\leq 2-\frac{1}{\gamma},$$
we obtain
$k_0/\gamma\leq 2k_0-\mu k_0< 2k_0-\mu k_0+\mu k_1 r_1.$ Consequently,
$$  \gamma> \dfrac{k_0}{2k_0-\mu(k_0-k_1 r_1)}; \qquad\alpha-\beta<0\qquad \alpha-\frac{\beta}{r_1}<0.$$
In particular, this implies that
$1>\gamma>\frac12,$ $Q_0\leq 0,$ $P_1<0.$
We now proceed to analyze the sign of $\mathrm{N}$ case by case.\\

\noindent (1) Suppose that $P_1 + 2\sqrt{(1-2\gamma)\gamma(1-\mu)Q_0} \ge 0.$ Then
$$(P_1 + 2\sqrt{(1-2\gamma)\gamma(1-\mu)Q_0})\varepsilon
- \gamma(1-\mu)Q_0 \ge 0,$$ and by hypothesis
$(P_1 - 2\sqrt{(1-2\gamma)\gamma(1-\mu)Q_0})\varepsilon- \gamma(1-\mu)Q_0 \ge 0.$
Consequently, $\mathrm{N} \ge 0$.\\

\noindent (2) Suppose that $P_1 + 2\sqrt{(1-2\gamma)\gamma(1-\mu)Q_0} < 0.$ By hypothesis,
\begin{equation}\nonumber
\varepsilon \le
\frac{-\gamma(1-\mu)Q_0}
{-P_1 + 2\sqrt{(1-2\gamma)\gamma(1-\mu)Q_0}}
<
\frac{-\gamma(1-\mu)Q_0}
{-P_1 - 2\sqrt{(1-2\gamma)\gamma(1-\mu)Q_0}}.
\end{equation}
Hence,
$(P_1 + 2\sqrt{(1-2\gamma)\gamma(1-\mu)Q_0})\varepsilon
- \gamma(1-\mu)Q_0 \ge 0$ and again by hypothesis,
\begin{equation}\nonumber
\left(P_1 - 2\sqrt{(1-2\gamma)\gamma(1-\mu)Q_0}\right)\varepsilon
- \gamma(1-\mu)Q_0 \ge 0.
\end{equation}
Thus $N\ge 0$ also in this case. Therefore the radical term is real, and both equilibria are well defined.
Next, let us consider \eqref{eq:simpParmetros}. By the previous assumptions,
$$\alpha\le 0,\qquad \alpha-\beta<0,\qquad \alpha-\frac{\beta}{r_1}<0,\qquad\delta<0.$$
Furthermore, from the bound on $\varepsilon$ we obtain
\begin{align*}
\varepsilon
&\le
\dfrac{\gamma(1-\mu)\big(k_1\gamma\mu r_1-k_0(1-2\gamma+\gamma\mu)\big)}
{k_1\gamma\mu-k_0(1-2\gamma+\gamma\mu)
+2\sqrt{(1-2\gamma)\gamma(1-\mu)\big(k_0(1-2\gamma+\gamma\mu)-k_1r_1\gamma\mu\big)}}\\
&\le
\dfrac{\gamma(1-\mu)\big(k_1\gamma\mu r_1-k_0(1-2\gamma+\gamma\mu)\big)}
{k_1\gamma\mu-k_0(1-2\gamma+\gamma\mu)}\\
&=
\frac{\delta(\alpha-\beta)}{\beta/r_1-\alpha}.
\end{align*}
Hence $\varepsilon(\beta/r_1-\alpha)\le \delta(\alpha-\beta),$ which is equivalent to
\begin{equation}\nonumber
   \delta(\alpha-\beta)+\varepsilon\left(\alpha-\frac{\beta}{r_1}\right)\ge 0. 
\end{equation}

We now show that each coordinate of $\mathbf{x}_4^*$ and $\mathbf{x}_5^*$ belongs to $[0,1]$. For $x_{0(4,5)}$, the previous inequality implies that the numerator is nonnegative, and since $k_0\delta(4\gamma-2)(\alpha-\beta)>0,$ we obtain $x_{0(4,5)}\ge 0.$ Moreover,
\begin{align*}
x_{0(4,5)}
&\le
\frac{\alpha\left[-\left(\delta(\alpha-\beta)+\varepsilon(\alpha-\beta/r_1)\right)
-\sqrt{\left(\delta(\alpha-\beta)+\varepsilon(\alpha-\beta/r_1)\right)^2-4\delta\varepsilon^2(2\gamma-1)(\alpha-\beta)}\right]}
{k_0\delta(4\gamma-2)(\alpha-\beta)}.
\end{align*}
Then, 
\begin{align*}
x_{0(4,5)}&<
\frac{2\alpha\left[-\left(\delta(\alpha-\beta)+\varepsilon(\alpha-\beta/r_1)\right)\right]}
{k_0\delta(4\gamma-2)(\alpha-\beta)}\\
&<
\frac{2\alpha\left[-(\delta+\varepsilon)(\alpha-\beta)\right]}
{k_0\delta(4\gamma-2)(\alpha-\beta)}\\
&=
\frac{2((2-\mu)\gamma-1)(\delta+\varepsilon)}{\delta(4\gamma-2)}
<
\frac{2(2\gamma-1)}{4\gamma-2}
=1.
\end{align*}
Therefore,$0\le x_{0(4,5)}<1.$
A similar argument applies to $x_{1(4,5)}$. Indeed, since
\[
0<
\delta(\alpha-\beta)+\varepsilon\left(\alpha-\frac{\beta}{r_1}\right)
<
\delta(\alpha-\beta)+\varepsilon\left(\alpha+\frac{\beta}{r_1}\right),
\]
we obtain $x_{1(4,5)}\ge 0,$ and
\begin{align*}
x_{1(4,5)}
&\le
\frac{\mu\left[\delta(\alpha-\beta)+\varepsilon(\alpha+\beta/r_1)
+\sqrt{\left(\delta(\alpha-\beta)+\varepsilon(\alpha-\beta/r_1)\right)^2-4\delta\varepsilon^2(2\gamma-1)(\alpha-\beta)}\right]}
{(4\gamma-2)(\mu-1)(\alpha-\beta)}\\
&\le
\frac{2\mu\big(\delta(\alpha-\beta)+\varepsilon\alpha\big)}
{(4\gamma-2)(\mu-1)(\alpha-\beta)}
\le
\frac{2\mu\delta(\alpha-\beta)}
{(4\gamma-2)(\mu-1)(\alpha-\beta)}\\
&=
\frac{2\mu\gamma}{4\gamma-2}
\le 1,
\end{align*}
where in the last inequality we used $\mu\le 2-\frac1\gamma$. Hence $0\le x_{1(4,5)}\le 1.$ For $p_{0(4,5)}$, using again the same sign information, we find $p_{0(4,5)}\ge 0,$ and
\begin{align*}
p_{0(4,5)}
&\le
\frac{\alpha\left[\delta(\alpha-\beta)-\varepsilon(\alpha-\beta/r_1)
+\sqrt{\left(\delta(\alpha-\beta)+\varepsilon(\alpha-\beta/r_1)\right)^2-4\delta\varepsilon^2(2\gamma-1)(\alpha-\beta)}\right]}
{2\delta(\alpha+(1-2\gamma)\varepsilon-\gamma k_1\mu)(\alpha-\beta)}\\
&\le
\frac{2\alpha\delta(\alpha-\beta)}
{2\delta(\alpha+(1-2\gamma)\varepsilon-\gamma k_1\mu)(\alpha-\beta)}\\
&=
\frac{\alpha}{\alpha+(1-2\gamma)\varepsilon-\gamma k_1\mu}
<1.
\end{align*}
Thus $0\le p_{0(4,5)}<1.$ Likewise, for $p_{1(4,5)}$ we have $p_{1(4,5)}\ge 0,$ and
\begin{align*}
p_{1(4,5)}
&\le
\frac{k_1\mu\left[(-\delta(\alpha-\beta)+\varepsilon(\alpha-\beta/r_1))
-\sqrt{\left(\delta(\alpha-\beta)+\varepsilon(\alpha-\beta/r_1)\right)^2-4\delta\varepsilon^2(2\gamma-1)(\alpha-\beta)}\right]}
{2(\mu-1)(\alpha+(1-2\gamma)\varepsilon-\gamma k_1\mu)(\alpha-\beta)}\\
&\le
\frac{2k_1\mu(-\delta)(\alpha-\beta)}
{2(\mu-1)(\alpha+(1-2\gamma)\varepsilon-\gamma k_1\mu)(\alpha-\beta)}\\
&=
\frac{-\gamma k_1\mu}{\alpha+(1-2\gamma)\varepsilon-\gamma k_1\mu}
<1.
\end{align*}
Hence $0\le p_{1(4,5)}<1.$
It remains to verify the simplex constraints. Using the explicit formulas and the estimates above, we obtain
\begin{align*}
0&<x_{0(4)}+x_{1(4)}\\
&\le
(k_0\gamma\mu-\alpha)
\frac{\delta(\alpha-\beta)+\varepsilon(\alpha+\beta/r_1)
+\sqrt{\left(\delta(\alpha-\beta)+\varepsilon(\alpha-\beta/r_1)\right)^2-4\delta\varepsilon^2(2\gamma-1)(\alpha-\beta)}}
{k_0\gamma(4\gamma-2)(\mu-1)(\alpha-\beta)}\\
&\le 1,
\end{align*}
and similarly
\begin{align*}
0&<p_{0(4)}+p_{1(4)}\\
&\le
(k_1\gamma\mu-\alpha)
\frac{(-\delta(\alpha-\beta)+\varepsilon(\alpha-\beta/r_1))
-\sqrt{\left(\delta(\alpha-\beta)+\varepsilon(\alpha-\beta/r_1)\right)^2-4\delta\varepsilon^2(2\gamma-1)(\alpha-\beta)}}
{2\gamma(\mu-1)(\alpha+(1-2\gamma)\varepsilon-\gamma k_1\mu)(\alpha-\beta)}\\
&\le 1.
\end{align*}
The same estimates hold for $\mathbf{x}_5^*$, the only difference being the choice of sign in front of the common radical term. Therefore, $\mathbf{x}_4^*,\mathbf{x}_5^*\in\mathcal D.$
\hspace*{\fill}$\square$\bigskip\noindent

\section{Proof of Proposition 6}

Since the coordinates of $\mathbf{x}_2^*$ and $\mathbf{x}_3^*$ share a common radical term, we first verify that this expression is real. It is therefore sufficient to examine the quantity under the square root, which can be written as
\begin{align*}
\mathcal{N}   &= \alpha \left[\alpha - 4 r_0 (1 - 2\gamma) \Delta \right] \varepsilon^2    -2 \alpha^2 \Delta r_0 \varepsilon + \alpha^2 \Delta^2 r_0^2 
\end{align*} 
with, $\alpha =k_0( 1 +\gamma (\mu-2))<0$ and $\Delta = \gamma( 1-\mu )>0.$
A direct factorization of the quadratic polynomial in $\varepsilon$ yields
\begin{align*}
\mathcal{N}
&= \alpha \left[ \alpha - 4\Delta r_0 (1 - 2\gamma) \right]
(\varepsilon - \varepsilon_-)(\varepsilon - \varepsilon_+),\qquad \varepsilon_{\pm}
= \frac{-\Delta r_0\alpha}{-\alpha \pm 2\sqrt{\alpha \Delta r_0 (1 - 2\gamma)}}.
\end{align*}
The sign of the denominator in the expressions of $\varepsilon_{\pm}$ leads to the following cases:
\begin{itemize}
    \item If $ -\alpha - 2  \sqrt{    \alpha \Delta r_0 (1 - 2\gamma) }>0$, then, under the assumptions on $\varepsilon$, we have $\varepsilon<\varepsilon_+<\varepsilon_-.$ 
    Consequently,
    $\alpha \left[ \alpha - 4\Delta r_0 (1 - 2\gamma) \right]>0, \quad  \left( \varepsilon - \varepsilon_- \right)<0,\quad \left( \varepsilon - \varepsilon_+ \right)<0.$
    Hence, $\mathcal{N} >0$.
    \item If $ -\alpha - 2  \sqrt{    \alpha \Delta r_0 (1 - 2\gamma) }<0$, then, under the assumptions on $\varepsilon$, it follows that the following inequalities hold: 
    $\alpha \left[ \alpha - 4\Delta r_0 (1 - 2\gamma) \right]<0, \quad  \left( \varepsilon - \varepsilon_- \right)>0,\quad \left( \varepsilon - \varepsilon_+ \right)<0.$
    Hence, $\mathcal{N} >0$.
    \item If $ -\alpha - 2  \sqrt{    \alpha \Delta r_0 (1 - 2\gamma) }=0$, then from the hypotheses it follows that $\varepsilon \leq \frac{\Delta r_0}{2}$, and consequently $\mathcal{N} = -2 \alpha^2 \Delta r_0 \varepsilon + \alpha^2 \Delta^2 r_0^2 \geq 0.$
  \end{itemize}
Under the stated assumptions, the equilibrium coordinates are
    real. It remains to show that each coordinate is positive. In particular, the hypotheses imply that
$$\varepsilon\leq  \Delta r_0,\quad\gamma  (\mu -2)+1<0,\quad \gamma>\frac{1}{2}.$$
Furthermore, using the assumptions stated above, we proceed to estimate each coordinate of the equilibria.
We begin with the equilibrium $\mathbf{x}_2^*$. For its first coordinate, by rationalizing the denominator we obtain
\begin{align*}
0 < x_{20}
&= \frac{-2\varepsilon^2\alpha}{k_0\big(\alpha(\varepsilon-\Delta r_0) + \sqrt{(\alpha(\varepsilon-\Delta r_0))^2 + 4\alpha \Delta r_0 \varepsilon^2 (2\gamma -1)}\big)} \\
&= \frac{-2\varepsilon^2\alpha\big(\alpha(\varepsilon-\Delta r_0) - \sqrt{(\alpha(\varepsilon-\Delta r_0))^2 + 4\alpha \Delta r_0 \varepsilon^2 (2\gamma -1)}\big)}{-4\alpha \Delta r_0 \varepsilon^2 (2\gamma -1)k_0}.
\end{align*}
Using that $\alpha<0$ and $\varepsilon \leq \Delta r_0$, we estimate
\begin{align*}
x_{20}
&\leq \frac{2\varepsilon^2\alpha\big(\alpha(\Delta r_0-\varepsilon)\big)}{-4\alpha \Delta r_0 \varepsilon^2 (2\gamma -1)k_0}
\leq\frac{\Delta r_0 - \varepsilon}{2\Delta r_0}
\leq \frac{1}{2}.
\end{align*}
A similar argument yields for the second coordinate
\begin{align*}
0 < x_{21}
&= \frac{2\varepsilon^2\gamma\mu}{\alpha(\varepsilon-\Delta r_0) + \sqrt{(\alpha(\varepsilon-\Delta r_0))^2 + 4\alpha \Delta r_0 \varepsilon^2 (2\gamma -1)}} \\
&\leq \frac{\gamma\mu(\Delta r_0-\varepsilon)}{2\Delta r_0(2\gamma-1)}
< \frac{\Delta r_0-\varepsilon}{2\Delta r_0}
\leq \frac{1}{2}.
\end{align*}

For the third coordinate, we obtain
\begin{align*}
0 < p_2
&= \frac{-2\varepsilon\alpha}{\alpha(-\varepsilon-\Delta r_0) + \sqrt{(\alpha(\varepsilon-\Delta r_0))^2 + 4\alpha \Delta r_0 \varepsilon^2 (2\gamma -1)}} \\
&< \frac{-2\varepsilon\alpha}{\alpha(-\varepsilon-\Delta r_0)}
= \frac{2\varepsilon}{\Delta r_0+\varepsilon}
\leq 1.
\end{align*}

Therefore,
\[
0 < x_{20},\, x_{21} < \frac{1}{2}, \qquad 0 < p_2 \leq 1,
\quad \text{and} \quad x_{20}+x_{21}<1.
\]

We now consider the equilibrium $\mathbf{x}_3^*$. Proceeding analogously, we obtain for the first coordinate
\begin{align*}
0 < x_{30}
&= \frac{-2\varepsilon^2\alpha}{k_0\big(\alpha(\varepsilon-\Delta r_0) - \sqrt{(\alpha(\varepsilon-\Delta r_0))^2 + 4\alpha \Delta r_0 \varepsilon^2 (2\gamma -1)}\big)} \\
&= \frac{-2\varepsilon^2\alpha\big(\alpha(\varepsilon-\Delta r_0) + \sqrt{(\alpha(\varepsilon-\Delta r_0))^2 + 4\alpha \Delta r_0 \varepsilon^2 (2\gamma -1)}\big)}{-4\alpha \Delta r_0 \varepsilon^2 (2\gamma -1)k_0} \\
&\leq \frac{   4\varepsilon^2\alpha\left( \alpha( \Delta r_0-\varepsilon) \right)}{- 4\alpha \Delta r_0 \varepsilon^2   (2\gamma -1)  k_0 }<\frac{\Delta r_0-\varepsilon}{\Delta r_0}
\leq 1.
\end{align*}

Similarly, for the second coordinate we obtain
\begin{align*}
0 < x_{31}
&= \frac{2\varepsilon^2\gamma\mu}{\alpha(\varepsilon-\Delta r_0) - \sqrt{(\alpha(\varepsilon-\Delta r_0))^2 + 4\alpha \Delta r_0 \varepsilon^2 (2\gamma -1)}} \\
&\leq \frac{\gamma\mu(\Delta r_0-\varepsilon)}{\Delta r_0(2\gamma-1)}
< \frac{\Delta r_0-\varepsilon}{\Delta r_0}
\leq 1.
\end{align*}

Finally, for the third coordinate we obtain
\begin{align*}
0 < p_3
&= \frac{-2\varepsilon\alpha}{\alpha(-\varepsilon-\Delta r_0) - \sqrt{(\alpha(\varepsilon-\Delta r_0))^2 + 4\alpha \Delta r_0 \varepsilon^2 (2\gamma -1)}} \\
&\leq \frac{  -2\varepsilon \alpha  ( \alpha(-\varepsilon-\Delta r_0) + \alpha(\varepsilon-\Delta r_0)   )}{ -4\alpha \Delta r_0 \varepsilon^2    (2\gamma -1) }\\
&= \frac{-\alpha}{2\gamma-1}
= \frac{2\gamma-1-\gamma\mu}{2\gamma-1}
< 1.
\end{align*}

Consequently,
$$0 < x_{30},\, x_{31} < 1, \qquad 0 < p_3 < 1,
\quad \text{and} \quad x_{30}+x_{31}\leq \frac{   4\varepsilon^2(2\gamma-1)    ( \alpha( \Delta r_0-\varepsilon)  )}{ 4\alpha \Delta r_0 \varepsilon^2    (2\gamma -1)}<1.$$

This proves that all coordinates of $\mathbf{x}_2^*$ and $\mathbf{x}_3^*$ are positive and satisfy the defining constraints of the region $\tilde{\mathcal D}$. Hence, both equilibria lie in $\tilde{\mathcal D}$.
\hspace*{\fill}$\square$\bigskip\noindent

\section*{Acknowledgments}
\noindent We want to thank J. Tom\'as L\'azaro for the careful reading of this manuscript and for his comments and suggestions.

\section*{Statements and Declarations}

\noindent\textbf{Funding:}  
ET was funded by the Chilean research agency ANID through the doctoral fellowship ANID/PhD/2021–21210522, and also received partial support from the Research Group in Ordinary Differential Equations and Applications (GIEDOAp), GI2310532-VRIP-UBB. SFE was supported by grants PID2025-169610NB-I00 (funded by Spain MCIU/AEI/10.13039/501100011033 and by “ERDF a way of making Europe”) and CIPROM/2022/59 (funded by Generalitat Valenciana). This work was also supported by the Spanish State Research Agency (AEI), through the Severo Ochoa and Mar\'ia de Maeztu Program for Centers and Units of Excellence in RD (CEX2020-001084-M). We thank CERCA Programme/Generalitat de Catalunya for institutional support. \\

\noindent\textbf{Author Contributions:}  
Model design: JS, SFE; Mathematical analyses: ET, FC, NM, JS; Conceptualization, analysis, and writing of the manuscript: all authors.\\

\noindent\textbf{Data Availability:}  
No datasets were generated or analysed during the current study.\\

\noindent\textbf{Ethics Approval:}  
This article does not contain any studies with human participants or animals performed by any of the authors.

\bibliographystyle{unsrt} %Referencias por orden de aparición
\bibliography{Bib_Delay_Quasispecies}

\end{document}